\documentclass[final]{siamltex}

\usepackage{amsmath,amsfonts,amssymb}

\usepackage{graphicx}
\usepackage{epstopdf}

\usepackage[margin=5pt,justification=centering]{subfig}

\usepackage{tabularx}

\usepackage{paralist}

\usepackage{braket}

\usepackage{tikz}
\usetikzlibrary{arrows,positioning,backgrounds}

\usepackage{wrapfig}

\usepackage[notref,notcite]{showkeys}

\usepackage[colorlinks,urlcolor=blue,citecolor=blue,linkcolor=blue]{hyperref}

\usepackage{algorithm}
\usepackage{algpseudocode}
\algnotext{EndProcedure}

\newcommand{\Sec}[1]{\hyperref[sec:#1]{\S\ref*{sec:#1}}} %
\newcommand{\App}[1]{\hyperref[sec:#1]{Appendix~\ref*{sec:#1}}} %
\newcommand{\Eqn}[1]{\hyperref[eq:#1]{(\ref*{eq:#1})}} %
\newcommand{\FigsTwo}[2]{{Figures~\ref{fig:#1} and~\ref{fig:#2}}} %
\newcommand{\Fig}[1]{\hyperref[fig:#1]{Figure~\ref*{fig:#1}}} %
\newcommand{\Tab}[1]{\hyperref[tab:#1]{Table~\ref*{tab:#1}}} %
\newcommand{\Thm}[1]{\hyperref[thm:#1]{Thm.\,\ref*{thm:#1}}} %
\newcommand{\Lem}[1]{\hyperref[lem:#1]{Lem.\,\ref*{lem:#1}}} %
\newcommand{\Prop}[1]{\hyperref[prop:#1]{Prop.~\ref*{prop:#1}}} %
\newcommand{\Cor}[1]{\hyperref[cor:#1]{Cor.~\ref*{cor:#1}}} %
\newcommand{\Def}[1]{\hyperref[def:#1]{Defn.~\ref*{def:#1}}} %
\newcommand{\Alg}[1]{\hyperref[alg:#1]{Algorithm~\ref*{alg:#1}}} %
\newcommand{\Ex}[1]{\hyperref[ex:#1]{Ex.~\ref*{ex:#1}}} %
\newcommand{\Clm}[1]{\hyperref[clm:#1]{Claim~\ref*{clm:#1}}} %
\newcommand{\Step}[1]{\hyperref[step:#1]{Step~\ref*{step:#1}}} %
\newcommand{\Line}[1]{\hyperref[line:#1]{Line~\ref*{line:#1}}} %
\newcommand{\Lines}[2]{\hyperref[line:#1]{Lines~\ref*{line:#1}--\ref*{line:#2}}} %

\newcommand{\qtext}[1]{\quad\text{#1}\quad}
\newcommand{\EX}{\mathbb{E}}
\newcommand{\Natural}{\mathbb{N}}

\newcommand{\ceil}[1]{\lceil #1 \rceil}

\newcommand{\ndseq}{\set{n_d}_{d \in \Natural}}
\newcommand{\cdseq}{\set{c_d}_{d \in \Natural}}
\newcommand{\nfillblk}{n^{\rm fill}_*}
\newcommand{\degblk}{d_*} 

\newcommand{\ndfill}[1][d]{n^{\rm fill}_{#1}}
\newcommand{\ndbulk}{n_d^{\rm bulk}}
\newcommand{\wdfill}[1][d]{w^{\rm fill}_{#1}}
\newcommand{\rdfill}[1][d]{r^{\rm fill}_{#1}}
\newcommand{\wdbulk}{w^{\rm bulk}_d}
\newcommand{\dmax}{d_{{\rm max}}}
\newcommand{\duniq}{d_{\rm uniq}}
\newcommand{\dmaxideal}{d^*}
\newcommand{\bmax}{b_{{\rm max}}}
\newcommand{\gmax}{g_{{\rm max}}}
\newcommand{\davg}{d_{{\rm avg}}}
\newcommand{\davgideal}{\bar d}
\newcommand{\Prob}[1]{\text{Pr}\left( #1 \right)}

\newcommand{\ColTitle}[1]{\multicolumn{1}{c|}{\bf #1}}

\begin{document}

\title{A Scalable Generative Graph Model\\with Community Structure\thanks{This work was funded by the GRAPHS Program at DARPA and by the Applied Mathematics Program at the U.S.\@ Department of Energy. Sandia National Laboratories is a multi-program laboratory managed and operated by Sandia Corporation, a wholly owned subsidiary of Lockheed Martin Corporation, for the U.S. Department of Energy's National Nuclear Security Administration under contract DE-AC04-94AL85000.}}
\author{%
  Tamara G. Kolda\footnotemark[2] 
  \and Ali Pinar\footnotemark[2] 
  \and Todd Plantenga\footnotemark[2] 
  \and C.\@ Seshadhri\footnotemark[2] 
}
\date{}
\maketitle

\renewcommand{\thefootnote}{\fnsymbol{footnote}}
\footnotetext[2]{Sandia National Laboratories, Livermore, CA.
  Email: \{tgkolda,apinar,tplante,scomand\}@sandia.gov}
\renewcommand{\thefootnote}{\arabic{footnote}}

\begin{abstract}
  Network data is ubiquitous and growing, yet we lack realistic
  generative network models that can be calibrated to match real-world data.
  The recently proposed Block Two-Level Erd\H{o}s-R\'{e}nyi (BTER)
  model can be tuned to capture two fundamental properties: degree
  distribution and clustering coefficients. The latter is particularly
  important for reproducing graphs with community structure, such as
  social networks. In this paper, we compare BTER to other scalable
  models and show that it gives a better fit to real data. We provide
  a scalable implementation that requires only O($\dmax$) storage
  where $\dmax$ is the maximum number of neighbors for a single
  node. 
  The generator is trivially parallelizable, and we show results
  for a Hadoop MapReduce implementation for a modeling a real-world web graph
  with over 4.6 billion edges. We propose that the BTER model can be
  used as a graph generator for benchmarking purposes and provide
  idealized degree distributions and clustering coefficient profiles
  that can be tuned for user specifications.
\end{abstract}

\begin{keywords}
  graph generator, network data, block two-level Erd\H{o}s-R\'{e}nyi (BTER)
  model, large-scale graph benchmarks
\end{keywords}

\section{Introduction}

Network interaction data are now available
from online social interactions, computer-to-computer communications,
financial transactions, collaboration networks, telecommunications,
and more.
A major obstacle to working in the field of network science is that
access to data is restricted due to a combination of security and
privacy concerns; yet models, algorithms, software, and hardware 
are struggling to keep pace with increasing demands for scalability 
and relevance. 
For these reasons, network science researchers need
scalable generative models for large-scale graphs.
Ideally, these generative models should capture salient features of
the networks being modeled.

Suppose we are given a graph representation of our data set.  We
introduce basic terminology for those unfamiliar with graph theory.
Let $G=(V,E)$ be an undirected, unweighted graph. We let $V$ denote
the set of \emph{vertices} or \emph{nodes} of the graph, and we let
$E$ denote the set of \emph{edges} where $(i,j) \in E$ means there is
an edge between nodes $i$ and $j$ or, equivalently, that nodes $i$ and
$j$ are \emph{adjacent}. We assume that the graph is simple, meaning
that the edges and undirected and unweighted and that there are no
self-edges.  Let $V_i = \set{j | (i,j) \in E}$ denote the set of
nodes that are adjacent to $i$ then the \emph{degree} of node $i$ is $d_i
= |V_i|$. The
\emph{transitivity} or \emph{clustering coefficient} of node $i$ is
defined as $c_i = |\set{(j,k) \in E | j,k \in V_i}| / \binom{d_i}{2}$, and it is a measure of the proportion of triangles that node
$i$ participates in compared to the number of possible triangles
\cite{WaSt98}. Roughly speaking, $c_i$ captures the social cohesion
around node $i$; for instance, in a clique (i.e., a graph in which
every node is adjacent to every other node), every node has $c_i = 1$.

In this paper, we consider how to reproduce two 
fundamental properties of graphs: the degree distribution and the
clustering coefficients by degree~\cite{ViBa12}. 
Let $V_d = \set{i | d_i =d}$ be the set of nodes of degree $d$ and
define $n_d = |V_d|$.  The degree distribution is specified by the
sequence $\ndseq$.  In most real-world networks
representing interaction data, there are a few nodes with high degrees
and many nodes with low degrees, with a smooth transition between. In
other words, the degree distribution is heavy-tailed, and this feature
has long been considered critical in distinguishing
real networks from arbitrary sparse
networks~\cite{BaAl99,ClShNe09,SaGaRoZh11}.
Let $c_d = \frac{1}{n_d} \sum_{i \in V_d} c_i$ be the average
clustering coefficient for nodes of degree $d$. The clustering
coefficients by degree are specified by the sequence $\cdseq$.  
The clustering coefficients of real-world graphs are much
higher than those of random graphs with the same degree distribution
\cite{GiNe02}.  Nonetheless, most generative models fail to match
clustering coefficients of real-world graphs~\cite{SaCaWiZa10}.

In previous work by a subset of the authors, we introduce the Block
Two-Level Erd\H{o}s--R\'enyi (BTER) \cite{SeKoPi12}.  The BTER model
can be tuned to capture both the degree distribution and degree-wise
clustering coefficients for real-world networks.  The inputs to the
BTER model are the sequences $\ndseq$ and $\cdseq$; our goal is to create a graph with similar behavior in these measures.  To do this, BTER works as
follows. The nodes are divided into \emph{affinity blocks}; see
\Fig{bin}. 
The fact that we can infer community structure from local clustering
coefficient measurements is justified in the original BTER paper
\cite{SeKoPi12} and now has additional theoretical justification
since it has recently been shown that triangle-rich graphs resemble
unions of dense blocks \cite{GuRoSe13}.
Each node has an assigned degree in such a way that, in
expectation, the degree distribution matches what it specified by
$\ndseq$. For each node, the links are divided between local links
within the affinity block (Phase 1 in \Fig{phase1}) and global links
that may connect to any node (Phase 2 in \Fig{phase2}). The proportion
of local links depends on the clustering coefficients specified by
$\cdseq$. Edges are generated by choosing endpoints in a
probabilistic way, as is explained in detail further on.
The goal of this paper is to explain the specifics of the model and
provide a scalable implementation. We also provide additional evidence
of the model's usefulness.

\begin{figure*}[thb]
  \centering
  \subfloat[Preprocessing: Distribution of nodes into affinity blocks]{\label{fig:bin}
    {\makebox[.32\textwidth]{\includegraphics[width=.3\textwidth,trim=0 0 0 0]{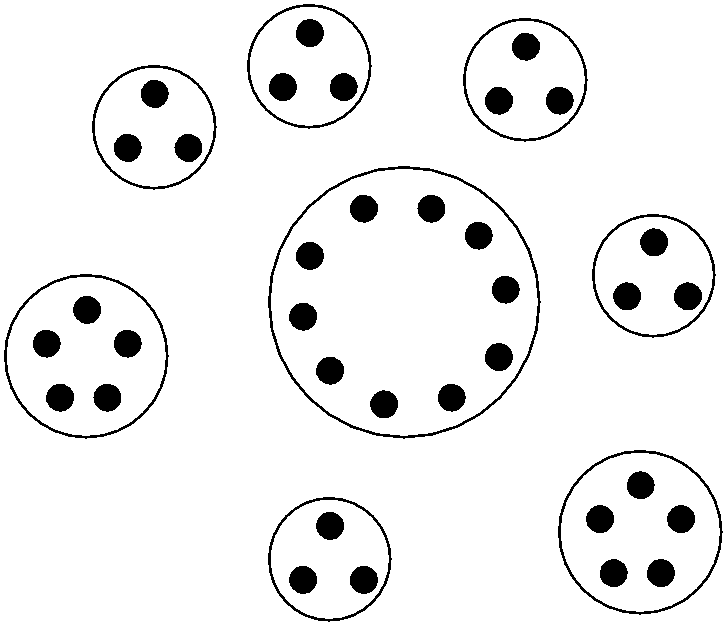}}}
  }
  \subfloat[Phase 1: Local links within each affinity block]{\label{fig:phase1}
    {\makebox[.32\textwidth]{\includegraphics[width=.3\textwidth,trim=0 0 0 0]{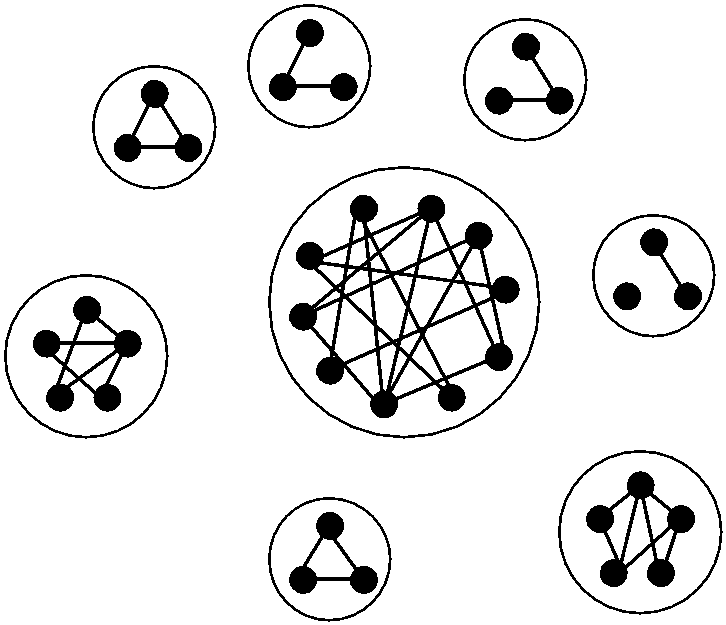}}}
  }
  \subfloat[Phase 2: Global links across affinity blocks]{\label{fig:phase2}
    {\makebox[.32\textwidth]{\includegraphics[width=.3\textwidth,trim=0 0 0 0]{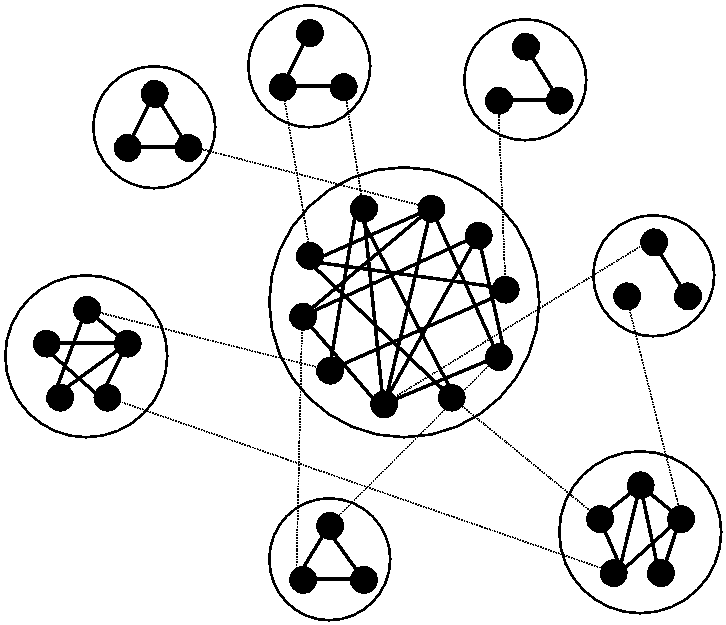}}}
  }
  \caption{BTER model phases \cite{SeKoPi12}.}
  \label{fig:bter}
\end{figure*}

\subsection{Contributions}

The BTER model \cite{SeKoPi12} was previously introduced as a scalable
model that can reproduce degree distributions and clustering
coefficients. The scalability is based on \emph{independent} edge
generation; i.e., there is no knowledge of previously-generated edges
in deciding on the next edge. However, the original paper
\cite{SeKoPi12} did not specify implementation details and moreover
recommended an  ad hoc procedure for some of the parameter
choices. 
In this paper, we make the following contributions:
\begin{asparaitem}
\item %
We provide a detailed reference implementation that clearly explains   how to
choose the BTER parameters to match the specified degree   distribution and
clustering coefficient profile. We note that there   is no iterative
optimization to fit BTER; the parameters of the   model are directly calculated
from the inputs.  The edges are   generated independently and in an arbitrary
order, so the BTER   generative model can potentially be used in streaming
scenarios.
\item %
We present efficient data structures for a scalable implementation,
requiring only $O(\dmax)$    storage and $O(\log \dmax)$ operations per edge, where $\dmax$ is the maximum degree.
Since our approach   generates all edges independently, it can be easily  
parallelized.  We provide examples demonstrating the scalability.
\item %
We demonstrate that BTER gives good approximations to the degree distributions and clustering coefficient behaviors of several large, heavy-tailed, real-world graphs.
We consider examples from the Laboratory for Web  
Algorithms~\cite{LWA}, including a graph with over 130 million nodes   and 4.6
billion edges,  the largest publicly
available graph that we are aware of.
We also compare BTER to competing methods on a pair of   smaller graphs.
\item %
We also consider how BTER may be used for arbitrary benchmarking   purposes,
when there is no target graph to match.  Since the model   requires a degree
distribution and clustering coefficient profile as   input, we focus on how to
generate these.  In particular, we   recommend the generalized log-normal
distribution for the degrees as   an alternative to the standard power law.
\end{asparaitem}

For very large graphs, the inputs can be expensive to compute,
especially the clustering coefficients. However, we have recently
proposed a sampling method that scales to very large graphs
\cite{SePiKo12,KoPiPlSe13}.

\subsection{Related Work}
\label{sec:related}
Since the goal of this paper is to focus on the implementation and
scalability of BTER, we limit our discussion to the most salient related models.
A more thorough discussion of related
work can be found in the paper that originally proposed the BTER model \cite{SeKoPi12}.

The majority of graph models add edges one at a time in a way that
each random edge influences the formation of future edges, making them
inherently serial and therefore unscalable. The classic example is Preferential Attachment
\cite{BaAl99}, but there are a variety of related models, e.g.,
\cite{KuRaRaSi00,LeKlFa07}. These models are more focused on capturing
qualitative properties of graphs and typically are difficult to match
to real-world data \cite{SaCaWiZa10}.  Perhaps the most relevant is
\cite{GuKr09}, which creates a graph with power law degree distribution and then ``rewires'' it
to improve the clustering coefficients.
The Musketeer model starts with a given graph and then does multilevel rewiring to attempt to preserve certain features of the original \cite{GuMeSa12}.
Another related model, the clustering model proposed by Newman
\cite{Ne03a}, assigns ``individuals'' to ``groups'' (a bipartite graph
with individual and group nodes) and then creates a graph of
connections between individuals by assigning connection probabilities
to each group; in other words, each group is modeled as an
Erd\H{o}s--R\'enyi graph.

A widely used model for modeling large-scale graphs is the Stochastic
Kronecker Graph (SKG) model, also known as R-MAT
\cite{ChZhFa04,LeChKlFa10}.  The generation process is easily
parallelized and can scale to very large graphs.  Notably, SKG has
been selected as the generator for the Graph 500 Supercomputer
Benchmark \cite{Graph500} and has been used in a variety of other
studies \cite{DoUrGiG10,MiBlWo10,MoKiNeVi10,MiWr12,GlOw12,FRD-arXiv-1210.5288,MiStBl12,Ke11}.  Unfortunately,
SKG has some drawbacks.
\begin{inparaenum}[(1)]
\item SKG can be extremely expensive to
fit to real data (using KronFit, the SKG parameter fitting algorithm proposed by the SKG inventors), and even then the fit is imperfect
\cite{LeChKlFa10}.
\item
It can generate only lognormal
tails (after suitable
addition of random noise)  \cite{JACM-SKG,SePiKo11}, limiting the
degree distributions that it can capture.
\item
Most importantly, SKG rarely closes wedges so the
clustering coefficients are much smaller than what is produced in real
data \cite{SaCaWiZa10,KoPiPlSe13}.
\end{inparaenum}

Another model of relevance is the Chung-Lu (CL) model \cite{ChLu02,
  ChLu02b, AiChLu01}. It is very similar to the edge-configuration
model of Newman et al.~\cite{NeWaSt02}. Let $d_i$ denote the
\emph{desired degree} for node $i$. In the Chung-Lu model, the probability
of an edge is proportional to the product of the degrees of its endpoints, i.e., the
probability of edge $(i,j)$ is $\propto d_i d_j$. Edges can be 
generated independently by picking endpoints proportional to their
desired degrees. If all degrees are the same, Chung-Lu reduces to the
well-known Erd\H{o}s--R\'enyi model \cite{ErRe60}.  The Chung-Lu model
is often used as a null model; for example, it is the basis of the
modularity metric
\cite{Ne06}. Graphs generated by the Chung-Lu model and the SKG model are, in fact, very similar
\cite{PiSeKo12}. The advantage of the Chung-Lu model is that it can be
better tuned to real-world degree distributions. The disadvantage of
the model is that, like SKG, it rarely closes wedges. 
Chung-Lu occurs as a special case of BTER when Phase 1 is skipped
(see \Sec{bter}).  The Chung-Lu construction is a very important part of BTER
and will be explained in more detail in the next section.

\section{Notation and background}
\label{sec:background}

Let $G=(V,E)$ be an undirected and unweighted graph, where $V$ denotes
the set of vertices and $E$ denotes the set of edges.
We let $n = |V|$ and $m=|E|$ denote the total number of vertices and
edges, respectively.
The set of node $i$'s neighbors and the degree of node $i$ are, respectively,
\begin{displaymath}
  V_i \equiv \set{j | (i,j) \in E} \qtext{and} d_i = |V_i|.
\end{displaymath}
The set of degree-$d$ nodes and the number of degree-$d$ nodes are, respectively,
\begin{displaymath}
  V_d = \set{ i | d_i = d } \qtext{and} n_d = |V_d|.
\end{displaymath}
The set $\set{n_d}$ defines the degree distribution.
Observe that the degree distribution specifies the total number of nodes and edges:
\begin{equation}
  \label{eq:nm}
  n = \sum_d n_d \qtext{and} m = \frac{1}{2} \sum_d d \cdot n_d.
\end{equation}
For convenience, we use the notation $\dmax = \max \set{ d_i | i \in V}$.

\subsection{Clustering coefficients}
We can discuss clustering coefficients in terms of wedges, closed
wedges, and triangles. A wedge is a path of length 2. \Fig{wedges}
shows a \emph{wedge} centered at node $j$, i.e., the path $i$-$j$-$k$
is a wedge.  We say wedge $i$-$j$-$k$ is closed if $(i,k) \in E$;
otherwise, the wedge is called open. A closed wedge forms a triangle,
i.e., a three-cycle.
\begin{figure}[th]
  \centering
    \begin{tikzpicture}[scale=0.3,
      nd/.style={circle,draw,fill=white,minimum size=4mm,inner sep=0pt}]
      \node (1) at (0,0) [nd] {$j$};
      \node (2) at (3,2) [nd] {$i$};
      \node (3) at (3,-2) [nd] {$k$};
      \draw (1) to (2);
      \draw (1) to (3);
      \draw [dashed] (2) to (3);
    \end{tikzpicture} 
  \caption{Wedge centered at node $j$. The wedge is closed if $(i,k)\in E$.}
  \label{fig:wedges}
\end{figure}
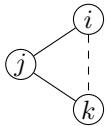
The number of wedges centered at node $i$ is $\binom{d_i}{2} = d_i(d_i-1)/2$. The clustering coefficient at node $i$ is
\begin{displaymath}
  c_i = \frac{\text{\# of closed wedges centered at node $i$}}{\text{\# wedges centered at node $i$}} = \frac{|\set{(j,k) \in E | j,k \in V(i)}|}{\binom{d_i}{2}}.
\end{displaymath}
This is the proportion of closed wedges divided by the total number of wedges.
For BTER, we are specifically interested in \emph{clustering coefficient per
  degree} \cite{WaSt98}, defined as
\begin{displaymath}
  c_d = \frac{\text{\# of closed wedges centered at a node of degree $d$}}{\text{\# wedges centered at a node of degree $d$}} 
  = \frac{1}{n_d} \sum_{i \in V_d} c_i.
\end{displaymath}
To summarize an entire graph, it is often convenient to consider the
\emph{global clustering coefficient} (GCC) \cite{WaSt98,BaWe00} , defined as
\begin{displaymath}
  c = \frac{\text{\# of closed wedges}}{\text{\# wedges}}.
\end{displaymath}
Note that this is \emph{not} the mean of the $c_i$'s. 
Since every triangle corresponds to three closed wedges, 
the number of closed wedges is three times the number of triangles.

\subsection{Erd\H{o}s--R\'enyi  Graphs}
As described in \Sec{bter}, BTER uses Erd\H{o}s--R\'enyi to model its affinity blocks.
An Erd\H{o}s--R\'enyi  graph
\cite{ErRe60} on $n$ vertices with connection probability $\rho$
is a graph such that each 
pair of vertices is independently connected with probability $\rho$. 
We refer to $\rho$ as the \emph{connectivity}.
If $\rho$ is
a constant, we call this a dense Erd\H{o}s--R\'enyi  graph; if $\rho = O(1/n)$, then
we call this a sparse Erd\H{o}s--R\'enyi  graph.

\subsection{Chung-Lu Graphs}
\label{sec:cl}

BTER uses a variation on Chung-Lu for its global links that span across
affinity blocks. Here we explain salient details that are relevant
to the implementation discussion.

The Chung-Lu model \cite{ChLu02, ChLu02b, AiChLu01}  approximates 
the probability that edge $(i,j)$ exists by
\begin{displaymath}
  \Prob{(i,j) \in E} = \frac{d_id_j}{2m},
\end{displaymath}
where $d_i$ indicates the desired degree of node $i$.
We can perform an independent coin flip for each of the possible $n^2$
edges, but this is too expensive.
We consider a ``fast'' version of the Chung-Lu model that generates
edges by choosing endpoints independently, proportional to their
degrees. Rather than doing an independent coin flip for each of $n^2$
possible edges, we can do $2m$ coin flips to pick two random endpoints
per edge. Each endpoint is picked independently such that the
probability of picking node $i$ is
\begin{displaymath}
  \Prob{i} = d_i / 2m.
\end{displaymath}
Hence, the probability (or, more precisely, the expected value) of
edge $(i,j)$ after picking $m$ edges is
\begin{displaymath}
   \Prob{(i,j) \in E} = m \cdot 2 \cdot \Prob{i} \cdot \Prob{j} = \frac{d_id_j}{2m}.
\end{displaymath}
The factor of 2 is because $(i,j)$ and $(j,i)$ are equivalent for an undirected graph.

Let $D_i$ correspond to the random variable that is the degree
of node $i$ for a single graph generated according to this
procedure. 
In the ``fast'' version, it is easy to see that the value of $D_i$ is Poisson distributed with mean $d_i$
\cite{FRD-arXiv-1210.5288}, so $\mathbb{E}(D_i) = d_i$, i.e., the expected value of
$D_i$ is $d_i$. 
Note that Chung-Lu admits non-integral desired degrees. If
the desired degree is $d_i=3.5$, then $\mathbb{E}(D_i) = 3.5$ even
though every realization $D_i$ is integral.

Despite the variation in individual degrees, the 
degree distribution is usually a good match to the desired one because there is spill-over from adjacent
degrees.  For instance, some nodes that were expected to become degree
4 are instead degree 5 and \emph{vice versa}. This will not be the
case if there are gaps in the degree distribution (i.e., $n_3, n_5 \gg
0$ and $n_4 = 0$). Also, degree-1 nodes may pose a particular problem.
According to the calculations in \cite{FRD-arXiv-1210.5288},
approximately 36\% of the pool of potential degree-1 nodes will not be
selected (i.e., have degree zero) and another 28\% will have degree 2 or
larger. To counteract these problems, \cite{FRD-arXiv-1210.5288}
proposes to increase the pool of potential degree-1 nodes while
keeping the total probability of potential degree-1 nodes
constant. 
Specifically, we ``blowup'' the set of degree-1 nodes by a factor $\beta \geq
1$. Hence, we add $(\beta -1)\cdot n_1$ nodes with desired degree
$d_1$, but we adjust the probability of picking any individual
degree-1 node so that
\begin{displaymath}
  \Prob{i} =
  \begin{cases}
    d_i / 2m & \text{if } d_i \geq 2,\\
    \beta^{-1}/2m & \text{if } d_i = 1.
  \end{cases}
\end{displaymath}
The BTER algorithm has a similar issue with degree-1 nodes, so the
reference implementation includes the option to specify a blowup
factor, $\beta$.

Finally, we note that since each endpoint and each edge is picked
independently, a graph generated according to fast Chung-Lu may contain
self-edges and repeat edges. We have ignored these details in the
discussion above because the practical impact has been small in our
experiments (we simply discard such edges).

\section{BTER Generative Graph Model}
\label{sec:bter}

We give a high-level overview of BTER in \Sec{bter:overview} and discuss the
challenges for a scalable implementation in \Sec{get-scale}. The remainder of
the section addresses the proposed solutions to these challenges and presents the implementation.

\subsection{The BTER Model} 
\label{sec:bter:overview}

BTER is based on the premise that a graph with a
heavy-tailed degree distribution and high clustering coefficients must contain
dense Erd\H{o}s--R\'enyi blocks; moreover, the distribution of the sizes of
those groups follows the same type of distribution of the degrees  \cite{SeKoPi12}.  Therefore,  BTER
organizes nodes into affinity blocks such that nodes within the same
affinity block have a much higher chance of being connected than nodes at
random, but BTER also behaves like the Chung-Lu model in that it is able to match an
arbitrary degree distribution.

The BTER model requires two user-specified inputs:
\begin{inparaenum}[(1)]
\item the desired degree distribution,
  $\ndseq$,  and
\item the desired clustering coefficients by degree, $\cdseq$.
\end{inparaenum}
These quantities may be measurements from an existing graph or set
arbitrarily, e.g., for benchmark purposes.
(We note that the original description in \cite{SeKoPi12} did not take
the original clustering coefficients but rather a function to
determine the connectivity.)
The desired number of nodes and edges can be computed from the degree distribution per \Eqn{nm}. 
There are three main steps to the graph generation, as described below.
The steps are depicted in \Fig{bter}.

\paragraph{Preprocessing} Imagine starting with $n$ isolated
vertices. Each vertex is assigned a degree, based on the degree
distribution $\{n_d\}$. So we arbitrarily assign $n_1$ vertices to
have degree $1$, $n_2$ to have to degree $2$, etc. We then partition
the $n$ vertices into affinity blocks.  For reasons that will be explained in the details that follow, an affinity block ideally
contains $d+1$ vertices that are assigned degree $d$. The idea is that each affinity block can potentially form a clique, in which case every node in that block has a clustering coefficient of 1.  
Note that there are many small blocks with
vertices of low degree, and a few large blocks of high degree
vertices.  At this stage, no edges have been added.

\paragraph{Phase 1} This phase adds edges \emph{within} each affinity
block. Each block is a dense Erd\H{o}s-R\'enyi graph, where the
density depends on the size of the block. For a block involving 
degree-$d$ vertices, the density is determined based on $c_d$. 
We show how to choose the connectivity within each block to ensure that
each vertex achieves its desired clustering coefficient 
and is not incident to more edges  than its desired
degree (in expectation).

\paragraph{Phase 2} This phase adds edges between the blocks. Consider
some vertex $i$ with an assigned degree $d_i$.  Suppose it is already
incident to $d'_i$ edges from Phase 1. We set $e_i = d_i - d'_i$ to be
the \emph{excess degree}\footnote{This definition of excess degree should not be confused with the ``excess degree distribution'' of Newman~\cite{Ne07}.} of $i$. We must create $e_i$ edges incident
to vertex $i$ to satisfy its degree requirement. We construct a
Chung-Lu graph with degree sequence $\set{e_i}_{i=1}^n$ to
complete the graph construction.

\subsection{Developing a scalable implementation} \label{sec:get-scale}

Our goal is to show that it is possible to have a highly scalable
implementation of the BTER method. 
The main goal is to have \emph{independent edge insertions} so that
the edge generation can be parallelized.

As stated, Phase 2 edge insertions must happen
after Phase 1, because we need to know the excess degrees. We parallelize this process
by computing the \emph{expected} excess degree. Given all the input parameters,
we can precompute the expected excess degree for any vertex (this requires compact
representations and data structures) during the preprocessing. 
From
this, we can precompute the total number of Phase~1 and Phase 2 edges. 

To perform a parallel edge insertion, we first decide randomly whether this should be
a Phase 1 or Phase 2 edge. For a Phase 1 edge, we select a
random affinity block (with the appropriate 
probability) and create an edge between two distinct randomly selected block members. For a Phase 2 edge, we perform
a Chung-Lu edge insertion based on expected excess degrees.
Because every edge is generated independently, there will be
duplicates, but these are discarded in the final graph.

Given the structure of parallel edge insertion, the main challenges in
developing a scalable implementation are as follows:
\begin{compactitem}
 \item \textbf{Preprocessing data structures.} A na\"ive implementation of the preprocessing
step would require $O(n)$ variables and storage by storing the Phase 1 and Phase 2 degree of every node. We design compact
representations and data structures for the affinity blocks. This 
contains all the relevant	information with minimal storage.
\item \textbf{Repeats in Phase 1.} Independent edge generation in 
Phase 1 leads to many repeated edges without enough distinct edges, 
and this affects the overall degree distribution when edge
repeats are removed. We show how to determine the
number of extra Phase 1 edges to be inserted
to rectify this.
\item \textbf{The degree-1 problem of Phase 2.} 
A parallel
implementation of Phase 2 results in numerous degree-1 
vertices becoming isolated. We use a fix for this discussed in \Sec{cl}, which is different than the one
proposed in the original BTER paper \cite{SeKoPi12}.
\end{compactitem}
Once these issues are addressed, we have an embarrassingly parallel edge generation algorithm that
requires only $O(\log \dmax)$ work per edge.
The remainder of this section gives an in-depth but informal presentation of our implementation. 
Detailed algorithm specifications  at pseudo-code level are also provided.  

\subsection{Preprocessing}
\label{sec:preprocessing}
Let $d_i$ denote the (desired) total degree of vertex $i$ and $w_i$ denote its (desired) excess degree.
For convenience, the nodes are indexed by increasing degree \emph{except} for
degree-1 nodes, which are indexed last. Hence, if
$d_i, d_j \geq 2$ and $i < j$, then $d_i \leq d_j$. As an example, see the numbering
in \Fig{example}.

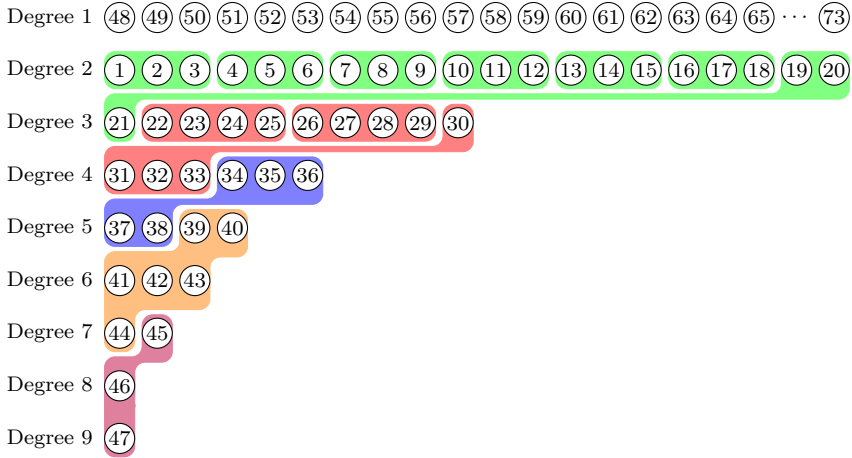
\begin{figure}[htb]
  \centering
%
%
%
\def\hv{0.5} %
\def\vv{9.5} %
\def\hb{0.1} %
\def\vb{0.35} %

\newcommand\funky[4]
{
    \filldraw[rounded corners]
    (\hb,\vv-#1-1) -- 
    (\hb,\vv-#1+\vb-0.8) -- 
    (#2-1+\hb,\vv-#1+\vb-0.8) -- 
    (#2-1+\hb,\vv-#1+\vb) -- 
    (#2+#3-1-\hb,\vv-#1+\vb) -- 
    (#2+#3-1-\hb,\vv-#1+\vb-0.9) --
    (#4-\hb,\vv-#1+\vb-0.9) --
    (#4-\hb,\vv-#1-1-\vb) --
    (\hb,\vv-#1-1-\vb) --
    (\hb,\vv-#1-1);
}

\newcommand\blocks[3]
{
  \foreach \s in {1,...,#3}
  {
    \filldraw[rounded corners] 
    (\s + \s*#1 - #1  - 1 + \hb + #2, \vv - #1 -\vb)
    rectangle +(#1 + 1 - 2*\hb, 2*\vb);
  }
}

\newcommand{\drawnodes}[3]
{
  \foreach \s in {#2,...,#3}
  {
    \node[vertex] at (\s-\hv-#2+1,\vv-#1) {\s};
  }
}

  \tikzstyle{vertex}=[draw,circle,fill=white,minimum size=4mm,inner sep=0]
  \begin{tikzpicture}[x=0.5cm,y=0.7cm]\footnotesize

    \begin{scope}[color=green!50!white]
      \blocks{2}{0}{6}
      \funky{2}{19}{2}{1}
    \end{scope}

    \begin{scope}[color=red!50!white]
      \blocks{3}{1}{2}
      \funky{3}{10}{1}{3}
    \end{scope}

    \begin{scope}[color=blue!50!white]
      \funky{4}{4}{3}{2}
    \end{scope}

    \begin{scope}[color=orange!50!white]
      \funky{5}{3}{2}{3}
      \funky{6}{1}{3}{1}
    \end{scope}

    \begin{scope}[color=purple!50!white]
      \funky{7}{2}{1}{1}
      \funky{8}{1}{1}{1}
    \end{scope}

    \drawnodes{1}{48}{65}
    \node at (19-\hv,\vv-1) {$\cdots$};
    \node[vertex] at (20-\hv,\vv-1) {73};
    \drawnodes{2}{1}{20}
    \drawnodes{3}{21}{30}
    \drawnodes{4}{31}{36}
    \drawnodes{5}{37}{40}
    \drawnodes{6}{41}{43}
    \drawnodes{7}{44}{45}
    \drawnodes{8}{46}{46}
    \drawnodes{9}{47}{47}

    \foreach \s in {1,...,9}
    {
      \node[left] at (0,\vv-\s) {Degree \s};
    }

  \end{tikzpicture} 
%
%
%

%
  \caption{Example of affinity blocks and groups for a graph of 73 nodes. Connected nodes denote the same affinity block. The same color denotes the same group membership.}
  \label{fig:example}
\end{figure}

In the preprocessing phase, we assign nodes to affinity blocks. We let $b_i$ denote the block assignment of node $i$.
For the assignment to affinity blocks, degree-1 nodes
are ignored.  The remaining nodes are assigned to
affinity blocks in order (of degree). A \emph{homogeneous} affinity block has
$d+1$ nodes of degree $d$. In \Fig{example}, the blocks are denoted by
colored ovals, and 4-5-6 is a
homogeneous block. The vast majority of (low-degree) nodes are
assigned to homogeneous affinity blocks. However, there are not always
enough nodes of degree $d$ to fill in a homogeneous block; therefore,
we also have a few (at most $\dmax$) \emph{heterogeneous} affinity
blocks with nodes of different degrees. For instance, in \Fig{example}, 19-20-21 is
a heterogeneous block.

For Phase 1, we divide the data into blocks. We let $i_b$ denote the starting index of block $b$, $d_b$ denote the minimum
degree of block $b$, $\rho_b$ denote its desired connectivity, $n_b$ denote the
number of nodes, $m_b$ denote
the desired number of edges in the affinity block, and $w_b$ denote the block \emph{weight} which is the number of edges to be inserted that accounts for expected repeats (see \Sec{phase2}).
Not all this information needs to be saved. To generate Phase 1 edges, we only need $i_b$, $n_b$, and $w_b$ to know the nodes that are involved with each block and how many edges to insert.
The process will insert $w_b$ edges which will, in expectation, produce $m_b$ unique edges so that each node within the block will have expected local (i.e., inside the affinity block) degree $m_b / {n_b \choose 2}$.
However, these three values are not what is stored since we can store fewer values for more efficiency.

We can compress the data further since all affinity blocks of the same size and
minimum degree can be grouped together into an \emph{affinity group} --- all
blocks in the same group share the same block size
and weight. 
In \Fig{example}, all nodes with the same color are in the same
affinity group, e.g., 1-21 are in the same affinity group, likewise
nodes 22-33, etc.
The information needed to store an affinity group boils
down to 4 items of information: $i_g$, the starting index of the
group; $b_g$, the number of blocks in the group; $n_g$, the size of each block; and
$w_g$, the total weight of all blocks in the group. The maximum number of groups is bounded by
$\dmax$, so we store at most $4\cdot\dmax$ values.

Phase 2 needs to build a Chung-Lu model using the excess degrees, $\set{e_i}$. We note that these values are computed in advance and are fixed through the entire process.
We can exploit the homogeneity of nodes of the same degree to save space for this calculation as well.
In a block where all nodes
have the same degree, we say the nodes are \emph{bulk nodes}.  In a
block with nodes of differing degrees, all nodes with degree equal to
the minimum degree are still bulk nodes. We refer to the remaining nodes as
 \emph{filler nodes}. In \Fig{example}, nodes 1-20 are degree-2
bulk nodes, nodes 22-30 are degree-3 bulk nodes, nodes 34-36 are
degree-4 bulk nodes, and so on. Node 21 is a degree-3 filler node, nodes
31-33 are degree-4 filler nodes, etc. It is possible to have either no
bulk nodes or no filler nodes for a given degree. In \Fig{example},
there are no filler degree-2 nodes and no bulk degree-6 nodes. 
Observe all bulk nodes of degree
$d$ (for any $d$) are in blocks of the same size and connectivity (i.e., same internal degree);
therefore, they all have the same excess degree. The filler nodes of
degree $d$ (for any $d$) participate in at most one block and so all
have the same excess degree. This means that there are two possible
values for excess degree for the set of nodes with desired degree
$d$. Hence, to generate Phase 2 links we need just 5 values for degree $d$: $\ndfill$ and $\ndbulk$, the number of
filler and bulk nodes; $\wdfill$ and $\wdbulk$, the total weight for filler nodes and for bulks nodes, which is exactly the total excess degree for each;
and the starting index $i_d$. Technically, the starting index can be
computed from $\set{n_d}$, but it reduces the work to store these indices
explicitly. However, we actually do not store $\ndbulk$ since it can be computed using $n_d$ and $\ndfill$. 
Additionally, because of how they are used, it is more convenient to store $w_d = \wdfill + \wdbulk$ and $r_d = \wdfill / w_d$ than $\wdfill$ and $\wdbulk$.
Hence, the total working storage for Phase 2 information
is $5\cdot\dmax$ values.

The total storage (including inputs) needed by the generation routine is $10
\cdot \dmax$ values (the sequence $\set{c_d}$ is not retained). It is possible to
modify the core data structures  to store only  the  distinct degrees instead of
maintaining a continuum of degrees through $\dmax$. This would change the
storage requirement to $O(\duniq)$ instead of $O(\dmax)$,  where $\duniq$ is the
number of distinct degrees in the graph.   However, we present our ideas  based
$O(\dmax)$ storage for clarity of presentation.   

For convenience, notation is described in \Tab{variables}. The node-
and block-level variables are not used in the algorithms.

\newcolumntype{m}{>{$}c<{$}}
\newcommand\grouptitle[2]{\multicolumn{4}{|c|}{\textbf{\em #1} $#2$}}
\begin{table}[htbp]
 \caption{Description of variables. A $\bigstar$ symbol in the second
  column indicates that this variable is explicitly computed and stored
  for use by the sampling procedure; a $\times$ symbol means the variable is not explicitly computed. The last column gives  the formula to derive it from
  the stored values.}
  \label{tab:variables}
  \centering\small
  \begin{tabular}{|m|m|l|l|}
    \hline
    \grouptitle{Scalars}{} \\ \hline
    n && Total number of nodes & $n = \sum n_d$\\ \hline
    m && Total number of edges & $m = \frac{1}{2} \sum d \cdot n_d$ \\ \hline
    w && Total number of edges to insert & $w = w^{(1)} + w^{(2)}$ \\ \hline
    \dmax & \bigstar & Largest (desired) degree &  \\ \hline
    \bmax & \times & Total number of affinity blocks & $\bmax = \sum_g b_g$ \\ \hline
    \gmax & \bigstar & Total number of affinity groups & $\gmax \leq  \dmax$ \\ \hline
    \beta & \bigstar & Blowup factor for degree-1 vertices & \\ \hline
    \grouptitle{Node Level}{i=1,\dots,n} \\ \hline
    d_i & \times & Degree (desired) of node $i$ & \\ \hline
    b_i & \times & Block id for degree $i$ & \\ \hline
    w_i & \times & Excess degree of node $i$ & $w_i = \frac{1}{2} [d_i -
(\rho_{b_i} \cdot d_{b_i})]$ \\ \hline
    \grouptitle{Block Level}{b=1,\dots,\bmax} \\ \hline
    d_b & \times & Minimum degree in block $b$ & \\ \hline
    \rho_b & \times & Connectivity of degree $b$ & $\rho_b = \sqrt[3]{c_{d_b}}$ \\ \hline
    n_b & \times & Number of nodes in block $b$ & $n_b = d_b + 1$ \\ \hline
    m_b & \times & Number of unique edges in block $b$ & $m_b = \rho_b \binom{n_b}{2}$ \\ \hline
    w_b & \times & Weight of block $b$ & $w_b = \binom{n_b}{2} \ln(1/(1-\rho_b))$ \\ \hline
    \grouptitle{Degree Level}{d=1,\dots,\dmax} \\ \hline
    n_d & \bigstar & Number of nodes of degree $d$ & \\ \hline
    c_d & \bigstar & Mean clustering coefficient for nodes of degree $d$ & \\ \hline
    i_d & \bigstar & Index of first degree of degree $d$ & \\ \hline
    n'_d & & Number of nodes of degree greater than $d$ & $n'_d = \sum_{d' > d} n'_d$ \\ \hline
    \ndfill & \bigstar & Number of fill nodes of degree $d$ & \\ \hline
    \ndbulk & & Number of bulk nodes of degree $d$ & $\ndbulk = n_d - \ndfill$ \\ \hline
    w_d & \bigstar & Excess degree of nodes of degree $d$ & \\ \hline
    r_d & \bigstar & Ratio of fill excess degree for degree $d$ & \\ \hline
    \wdfill & & Excess degree of fill nodes of degree $d$ & $\wdfill = r_d \cdot w_d$ \\ \hline
    \wdbulk & & Excess degree of bulk nodes of degree $d$ & $\wdbulk - w_d - \wdfill$ \\ \hline
    \grouptitle{Group Level}{g=1,\dots,\gmax} \\ \hline
    i_g & \bigstar & Index of first node in group $g$ & \\ \hline
    b_g & \bigstar & Number of affinity blocks in group $g$ & \\ \hline
    n_g & \bigstar & Number of nodes per block in group $g$ & \\ \hline
    w_g & \bigstar & Weight of group $g$ (including duplicate edges) & \\ \hline
    \grouptitle{Phase Level}{k=1,2} \\ \hline
    w^{(k)} & & Weight of phase $k$ & $w^{(1)} = \sum_g w_g$\\ 
    & & & $w^{(2)} = \sum w_d$ \\ \hline
  \end{tabular}
\end{table}

\subsection{Preprocessing Algorithm}
The BTER Setup procedure is described in \Alg{setup}. The inputs are the degree
distribution, $\set{n_d}$; the clustering coefficients per degree,
$\set{c_d}$; and the blowup factor for degree-1 nodes, $\beta$.

The method precomputes the index for the first node of each degree,
$\set{i_d}$, and the number of nodes with degree greater than the degree
$d$, $\set{n'_d}$. The latter is not saved after the preprocessing phase.

The degree-1 nodes are handled specially.  All degree-1 nodes are designated as
``fill'' nodes. The number of  candidate degree-1
nodes may be increased using the blowup factor, $\beta$. However, if
the blowup is used, the majority of the (desired) degree-1 nodes will
ultimately have degree 0 and can be removed in post-processing.

The main loop walks through each degree, determining the information
for Phases 1 and 2. 

It first allocates degree-$d$
nodes to
be fill nodes for the last incomplete block, if needed. The number
of nodes necessary to  complete the last incomplete block in the previous group is denoted by
$\nfillblk$. The excess degree of any fill nodes depends on the
internal degree of the last incomplete block, denoted by
$\degblk$. The excess degree is used to determine the weight of the
degree-$d$ fill nodes for phase 2, $\wdfill$.

If more nodes of degree $d$ remain, these are allocated as bulk
nodes and a new group is formed. 
The number of bulk nodes of degree $d$ is denoted by $\ndbulk$.
For the new group, we determine the index of the first node, the
number of blocks, and the size of each block.
The very last block of the very
last group is special because the remaining nodes may not be enough to fill it. For simplicity and because it is often the case for
heavy-tailed networks, we assume the last group contains only a one
block. This makes handling it as a special case simpler.
We compute
excess degree for these bulk nodes and then the corresponding weight
of degree-$d$ bulk nodes for Phase 2, $\wdbulk$. 
We also compute the weight of the group, $w_g$ according to a special formula as justified in \Sec{phase2}. 
Finally, we compute the number of nodes needed to fill out the final block in the group currently
being processed, $\nfillblk$.

Rather than storing $\wdfill$ and $\wdbulk$ directly, it is easier (for
the edge generation phase) to
have their sum, $w_d$, and the ratio of fill nodes, $r_d$.
Likewise, we do not return $\ndbulk$ since it can be easily recomputed
using $n_d$ and $\ndfill$. We do return $i_d$, but this could be
omitted and recomputed if that were more efficient (e.g., reducing
communication to workers in a parallel setting). Finally, we no longer need to keep
$\set{c_d}$ after the preprocessing is complete.

\begin{algorithm}[p]
  \caption{BTER Setup}
  \label{alg:setup}
  \begin{algorithmic}[1]\small
    \Procedure{bter\_setup}{$\set{n_d}$, $\set{c_d}$, $\beta$}
    \Statex
    \Statex \emph{Number nodes from least degree to greatest, except degree-1 nodes are last}
    \State $i_2 \gets 1$
    \For{$d=3,\dots,\dmax$}
    \State $i_d \gets i_{d-1} + n_{d-1}$
    \EndFor
    \State $i_1 \gets i_{\dmax} + n_{\dmax}$
    \Statex
    \Statex \emph{Compute number of nodes with degree greater than $d$}
    \For{$d=1,\dots,\dmax$}
    \State $n'_d \gets \sum_{d' > d} n'_d$ 
    \EndFor

    \Statex
    \Statex \emph{Handle degree-1 nodes}
    \State \label{line:blowup}
    $\ndfill[1] \gets \beta \cdot n_1$, $w_1 \gets \frac{1}{2} n_1$, $\rdfill[1] \gets 1$
    \Statex

    \Statex \emph{Main loop}
    \State $g \gets 0$, $\nfillblk \gets 0$, $\degblk \gets 0$
    \For{$d =2,\dots,\dmax$} 
    \If{$\nfillblk > 0$}
    \Comment{Try to fill incomplete block from current group}
    \State $\ndfill \gets \min(\nfillblk, n_d)$  
    \State $\nfillblk \gets \nfillblk - \ndfill$
    \State $\wdfill \gets \frac{1}{2} \ndfill  (d - \degblk)$
    \Else
    \State $\ndfill \gets 0$, $\wdfill \gets 0$
    \EndIf
    \State $\ndbulk \gets n_d - \ndfill$
    \If{$\ndbulk > 0$} 
    \Comment{Create a new group for degree-$d$ bulk nodes}
    \State $g \gets g + 1$
    \State $i_g \gets i_d + \ndfill$
    \State $b_g \gets \ceil{\ndbulk / (d+1)}$ 
    \State $n_g \gets d+1$ 
    \If{$b_g \cdot (d+1) > (n'_d + \ndbulk)$} 
    \Comment{Special handing of last group}
    \State \textbf{if} {$b_g \neq 1$} \textbf{then throw error}
    \State $n_g \gets (n'_d + \ndbulk)$ 
    \EndIf
    \State $\rho_* \gets \sqrt[3]{c_d}$
    \State $\degblk \gets (n_g - 1) \cdot \rho_*$
    \State $\wdbulk \gets \frac{1}{2} \ndbulk \cdot (d - \degblk)$
    \State $w_g \gets b_g \cdot \frac{1}{2} n_g(n_g-1) \cdot \log(1/1-\rho_*)$ \label{line:setup-wl}
    \State $\nfillblk \gets (b_g \cdot n_g) - \ndbulk$
    \Else
    \State $\wdbulk \gets 0$
    \EndIf
    \State $w_d \gets \wdfill + \wdbulk$, $r_d \gets \wdfill / w_d$

    \EndFor

    \Statex

    \State \Return $\set{i_d}, \set{w_d}, \set{r_d}, \set{\ndfill}, \set{w_g}, \set{i_g}, \set{b_g}, \set{n_g}$
    
    \EndProcedure

  \end{algorithmic}
\end{algorithm}

\subsection{Phase 1}
\label{sec:phase1}

Phase 1 creates intra-block links.    Each affinity block is modeled as
an Erd\H{o}s--R\'enyi graph. An overwhelming majority of the triangles are formed in this phase, and thus we pick the Erd\H{o}s--R\'enyi constant, $\rho$, for the block to match the target clustering coefficient   $c$.  A vertex of degree $d$ and clustering coefficient $c$  is incident to $c \cdot \binom{d}{2}$ triangles.   Assume this vertex is  grouped with other vertices of degree $d$ into a block with $d+1$ vertices, which holds  for all homogeneous blocks.  If we build an Erd\H{o}s--R\'enyi  graph of this block with parameter $\rho$, then this vertex is expected to be incident to $\binom{d}{2}\rho^3$ triangles.  Solving for $\rho$ yields $\rho = \sqrt[3]{c}$. 
Therefore, for block $b$, the connectivity  is $\rho_b = \sqrt[3]{c_{d_b}}$ where $d_b$ denotes the minimum degree in the block (since most blocks are homogeneous, this choice works well). Note that the  clustering coefficients  of vertices  will be higher  if we only consider  the affinity blocks. This is to compensate for the edges  that will be  added in Phase 2 to increase  the number of wedges, likely without contributing any triangles.

The difficulty in Phase 1 is that we expect a preponderance of repeat
edges because edges are generated independently.  Consider affinity block $b$ with $n_b$ nodes and
connectivity $\rho_b$, meaning that each node in block $b$ wants
internal degree $\rho_b \cdot d_b$. BTER wants approximately $m_b =
\rho_b \binom{n_b}{2}$ \emph{distinct} edges in block $b$.  Determining the number of draws
with replacement to get a desired number of distinct items can be cast as 
a \emph{coupon collector problem}. Specifically, the coupon collector problem is as follows: ``Suppose we have a box with $x$ distinct coupons. We draw a coupon and return it to the box (sampling with replacement). How many draws do we need to find $y$ distinct coupons?'' Our problem is slightly different than the standard problem since we want to draw $y$ distinct edges where $y = m_b$ may not be integral; this is fine since the goal is to achieve $m_b$ distinct edges in expectation.
We have developed a good approximation for the expected number of 
edges that need to be inserted:
\begin{equation}\label{eq:cc}
  w_b = \binom{n_b}{2} \ln(1/(1-\rho_b)).
\end{equation}
The proof is provided in \App{coupon_collector}.

We illustrate the utility of equation \Eqn{cc} with an example with $n_b=10$ nodes and
connectivity $\rho_b=0.5$, corresponding to $m_b=22.5$ edges, on
average. In this case, the formula predicts that we need to do
$w_b=31.1916$ draws, in expectation, to see the desired number of
unique edges, in expectation.
We do 10,000
random experiments as follows. For $i=1,\dots,10,000$, the random
variable $X_i \sim \text{Poisson}(w_b)$ is the number of items \emph{drawn}
from the $\binom{n_b}{2} = 45$ possible edges, and $Y_i$ is the number of
those items that are \emph{unique}. A histogram of the $Y_i$ values is shown
in \Fig{coupon_collector}. The average number of unique items is exactly  the desired value.
\begin{figure}[htbp]
  \centering
  \includegraphics[height=2.5in]{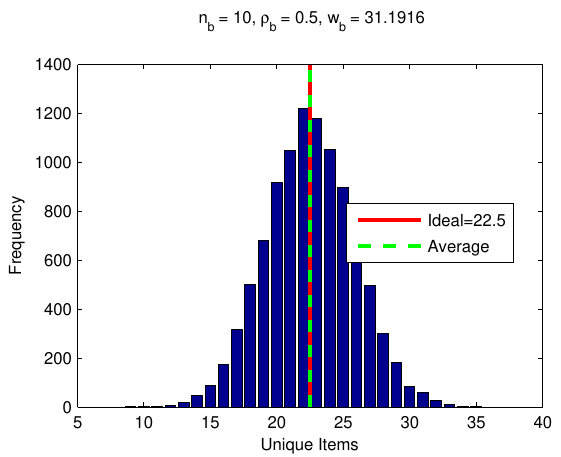}
  \caption{Distribution of the number of unique edges on 10,000 random
  trials.}
  \label{fig:coupon_collector}
\end{figure}

From equation \Eqn{cc}, we can determine the number of extra edges needed in Phase 1. Specifically, we
insert $w^{(1)} = \sum_b w_b$ edges to get a total of
$m^{(1)} = \sum_b m_b$ unique Phase 1 edges.
We let $w_g$ be the sum of the weights of all its constituent blocks.
For generating a Phase 1 edge, the process for a single edge proceeds as follows:
\begin{compactenum}
  \item Pick an affinity group randomly, where the probability of picking group $g$ is proportional to its weight $w_g$. 
  \item Pick a block within the affinity group uniformly at random
    (all blocks within the same group have the same weight).
  \item Pick two nodes  uniformly at random \emph{without replacement}\footnote{The terminology ``without replacement'' means that the first node is selected uniformly at random from the $n_b$ nodes in block $b$ and the second nodes is selected uniformly at random from the \emph{remaining} $n_b-1$ nodes in block $b$.}
    from the selected block --- these two nodes form an edge.
\end{compactenum}
The first step is a weighted sampling step and requires $O(\log \gmax)$ work,
where $\gmax \leq \dmax$ is the total number of affinity groups. The
other 2 steps are constant time operations.

\subsection{Phase 2}
\label{sec:phase2}
Phase 2 is simply applying the  Chung-Lu model on the expected excess degrees.
In creating an edge, we choose two nodes
independently. Those nodes are chosen proportional to their excess
degree. For node $i$ in group $b$, let $w_i = \frac{1}{2} [d_i -
(\rho_b \cdot d_b)]$ denote half its excess 
degree. The total number of edges that should be inserted
in Phase 2 is $w^{(2)} = \sum_i w_i$. 
Ignoring duplicate edges (which are fairly rare in our experiments), we have $m^{(2)} = w^{(2)}$ Phase 2 edges.

Let $\ndfill$ and $\ndbulk$ be the number of filler and bulk nodes of
degree $d$,
let $w_d = \sum_{i \in V_d} w_i$ be the weight of all degree-$d$ nodes,
and let $r_d$ be the proportion that are filler nodes.
Inserting a Phase 2 edge proceeds as follows:
\begin{compactenum}
\item Pick degree $d$, where the probability of picking degree $d$ is proportional to $w_d$.
\item Pick between filler and bulk nodes by selecting a uniform random number $x \in [0,1]$ and choosing filler if $x < r_d$ and bulk otherwise.
\item Pick a filler or bulk node (as determined in Step 2)
  of degree $d$, uniformly at random.
\end{compactenum}
The first step is a weighted sampling step and required
$O(\log(\dmax))$ work, while the other two steps are constant time operations.

One complication in Phase 2 is that getting the correct number
of degree-1 nodes poses a problem --- approximately 36\% of the pool of
potential degree-1 nodes will not be selected and another 28\% will
have degree 2 or larger. A fix for this problem has been proposed in
\cite{FRD-arXiv-1210.5288}, which involves increasing the pool of
degree-1 nodes, without changing the  the expected number of edges  that will be connected to these vertices.  This increase in the  pool size is  controlled by  the ``blowup factor,''  $\beta \geq 1$.  This
is included in the setup described in \Alg{setup} (\Line{blowup}).

\subsection{Independent Edge Generation}
\label{sec:dedup}

Lastly, we pull everything together to explain the independent edge
generation.  We insert a total of $w = w^{(1)} + w^{(2)}$ edges,
flipping a weighted coin for each edge to determine if it is Phase 1
or Phase 2. We expect to generate a total of $m = m^{(1)} + m^{(2)}$ edges.

Generating the edges is extremely inexpensive: $O(\log(\dmax))$ per
edge. The expensive step is de-duplication where extra copies of repeated edges are removed. The same difficulty exists
for the current Graph 500 (SKG) benchmark. Some may argue that
duplicate edges are a useful feature since real data also has
duplicates, but it is not clear that the 
duplication rates are similar to those observed in real data. 

\subsection{BTER Sampling Implementation}
\label{sec:implementation}

BTER edge generation is shown in \Alg{sample}. 
The procedure \textsc{Random\_Sample} does a weighted sampling
according to a specified discrete distribution. For $p$ bins, the cost
is $O(\log(p))$.
For each edge, we
randomly select between the phases using a weighted coin. 
A Phase 1 edge requires one sample from a discrete distribution of size
$\gmax$ and three additional random values drawn uniformly from
$[0,1]$. A Phase 2 edges requires two samples from a discrete distribution of size
$\dmax$ and four additional random values drawn uniformly from
$[0,1]$. Since $\gmax \leq \dmax$, an upper bound on the cost per edge is the cost of
one discrete random sample on a distribution of size $\dmax$ plus four
random values drawn uniformly from
$[0,1]$.

In \Alg{sample}, we generate each edge independently. It
may also be possible to ``bulk'' the computations by first determining
the total number of edges for each phase and perform the computation
for each phase separately. Within each phase, the procedure itself can
be easily vectorized  to boost runtime performance, as in MATLAB.

\begin{algorithm}[htbp]
  \caption{BTER Sample}
  \label{alg:sample}
  \begin{algorithmic}[1]\small
    \Procedure{bter\_sample}{$\set{n_d}, \set{i_d}, \set{w_d},
      \set{r_d}, \set{\ndfill}, \set{w_g}, \set{i_g}, \set{b_g},
      \set{n_g}$} 
    \State $w^{(1)} \gets \sum_g w_g$, 
    $w^{(2)} \gets \sum w_d$, $w \gets w^{(1)} + w^{(2)}$
    \State $E \gets \emptyset$
    \For{$j=1,\dots,w$}
    \State $r \sim U[0,1]$
    \If{$r < w^{(1)}/w$}
    \State $E \gets E \; \cup $ \textsc{bter\_sample\_phase1}($\set{w_g}, \set{i_g}, \set{b_g},
      \set{n_g}$)
    \Else
    \State $E \gets E \; \cup $ \textsc{bter\_sample\_phase2}($\set{w_d}, \set{r_d}, \set{n_d},
      \set{\ndfill}, \set{i_d}$)
    \EndIf
    \EndFor
    \State \Return $E$
    \EndProcedure
    \Statex

    \Procedure{bter\_sample\_phase1}{$\set{w_g}, \set{i_g}, \set{b_g},
      \set{n_g}$}
    \State $g \gets $ \textsc{random\_sample}($\set{w_g}$)
    \label{line:sample-l}
    \Comment{Choose group}
    \State $r_1 \sim U[0,1]$,
    $\delta = i_g + \lfloor r_1 \cdot b_g \rfloor \cdot n_g$
    \label{line:sample-delta}
    \Comment{Choose block and compute its offset}
    \State $r_2 \sim U[0,1]$,
    $i \gets \lfloor r_2 \cdot n_g \rfloor + \delta$
    \label{line:sample-i}
    \Comment{Choose 1st node}
    \State $r_3 \sim U[0,1]$,
    $j \gets \lfloor r_3 \cdot (n_g-1) \rfloor + \delta$
    \label{line:sample-ii}
    \Comment{Choose 2nd node}
    \If{$j \geq i$}
    \State $j \gets j + 1$
    \EndIf \label{line:sample-ii-end}
    \State \Return $(i,j)$
    \EndProcedure
    \Statex

    \Procedure{bter\_sample\_phase2}{$\set{w_d}, \set{r_d}, \set{n_d},
      \set{\ndfill}, \set{i_d}$}
    \State $i \gets $ \textsc{bter\_sample\_phase2\_node}($\set{w_d}, \set{r_d}, \set{n_d},
      \set{\ndfill}, \set{i_d}$)
    \State $j \gets $ \textsc{bter\_sample\_phase2\_node}($\set{w_d}, \set{r_d}, \set{n_d},
      \set{\ndfill}, \set{i_d}$)
    \State \Return $(i,j)$
    \EndProcedure
    
    \Statex 

    \Procedure{bter\_sample\_phase2\_node}{$\set{w_d}, \set{r_d}, \set{n_d},
      \set{\ndfill}, \set{i_d}$}
    \State $d \gets $ \textsc{random\_sample}($\set{w_d}$)
    \Comment{Choose degree}
    \State $r_1 \sim U[0,1]$, $r_2 \sim U[0,1]$
    \If{$r_1 < r_d$}
    \State $i \gets \lfloor r_2 \cdot \ndfill \rfloor + i_d $
    \Comment{Fill node}
    \Else
    \State $i \gets \lfloor r_2 \cdot (n_d - \ndfill) \rfloor + (i_d +
    \ndfill$)
    \Comment{Bulk node}
    \EndIf
    \State \Return $i$
    \EndProcedure
  \end{algorithmic}
\end{algorithm}

\subsection{Edge Deduplication}
Any method can be used for deduplication. In general, the simplest
procedure is to hash the edges in such a way that $(i,j)$ and $(j,i)$
hash to the same key.  Then it is easy enough to sort each bucket to
remove duplicates. In a parallel environment, since we are hashing by
edge and not vertex, there should not be load balancing problems. In
fact, hashing by a single endpoint is not recommended because of the
heavy-tailed nature of the graph.

\subsection{Implementations}
\label{sec:implementations}
We have a reference implementation in MATLAB that is available at \url{http://www.sandia.gov/~tgkolda/feastpack/}.
We have also implemented the method in
Hadoop MapReduce and use this version in some of our experiments. 
Since the BTER algorithm generates edges independently, map tasks can perform
all the work of creating edges and reduce tasks simply remove duplicate edges;
hence, the implementation runs as a single MapReduce job.
Each map task is given the desired degree and clustering coefficient, which is
sufficient to compute affinity blocks and sampling probabilities.  Each map
task uses a different seed for random number generation used in creating edges.
Map tasks are assigned a fixed number of edges to generate.  The default is
one million edges, so a graph of $w$ edges requires $w / 10^6$ map tasks.
There is no HDFS input file for the map tasks; instead, we wrote a subclass of
$\textsc{org.apache.hadoop.mapreduce.InputFormat}$
that causes a map task to generate a given number of records.
For each edge, the map emits a key-value pair consisting of a hash value for
the edge (based on the two endpoints) and the endpoints.  Note that there is no
assignment of specific nodes to specific map or reduce tasks.
The reducer tasks collect edges with the same hash value and remove any duplicates
before emitting the final list of all edges.  We enable Hadoop compression
between the map and reduce phases for faster performance.

\section{Numerical Comparisons}
\label{sec:realdata}

We consider the performance of BTER on various real-world data sets, including
presently the largest publicly-available graph  modeled in the sense of matching
degree distribution. We provide additional plots in the appendices; specifically, log-binned data as well as cumulative degree distributions.

\subsection{Small data}

In \Fig{smalldata}, we present comparisons  of BTER with the state-of-the-art in scalable generative models: SKG  \cite{ChZhFa04,LeChKlFa10} and Chung-Lu~\cite{AiChLu01,ChLu02, ChLu02b,
  FRD-arXiv-1210.5288} on two small data sets available from
SNAP~\cite{SNAP}. 
We use SKG because it is the current choice as the Graph 500 generator \cite{Graph500}. We omit details here but instead refer the reader to the references listed above for the details of the method and to  \cite{JACM-SKG,SePiKo11} for  and in-depth  analysis and discussion of problems with SKG.
The Chung-Lu model has been described in \Sec{cl}.
We treat all edges as
undirected and remove any duplicate edges and loops.
The graph ca-AstroPh is  a collaboration
network based on 124 months of data from the astrophysics section of
the arXiv preprint server; it has 9,987 nodes and 25,973 edges. 
The graph soc-Epinions1 is a who-trusts-whom online social (review) network from the Epinions
website with 75,879 nodes and 405,740 edges.

The parameters of SKG are from \cite{LeChKlFa10}.
The input to Chung-Lu is the degree distribution of the real graph 
and a blowup factor of $\beta=10$ \cite{FRD-arXiv-1210.5288}.
The inputs to BTER are the degree distribution and clustering
coefficients per degree (computed exactly by counting all wedges and triangles)
of the real graph and a blowup factor of $\beta=10$.
Timings are not reported as they are negligible  for all three methods.

\begin{figure}[htbp]
  \centering
  \subfloat[Degree distribution for ca-HepTh]
  {\label{fig:ca-HepTh-dd}\includegraphics[width=0.33\textwidth]{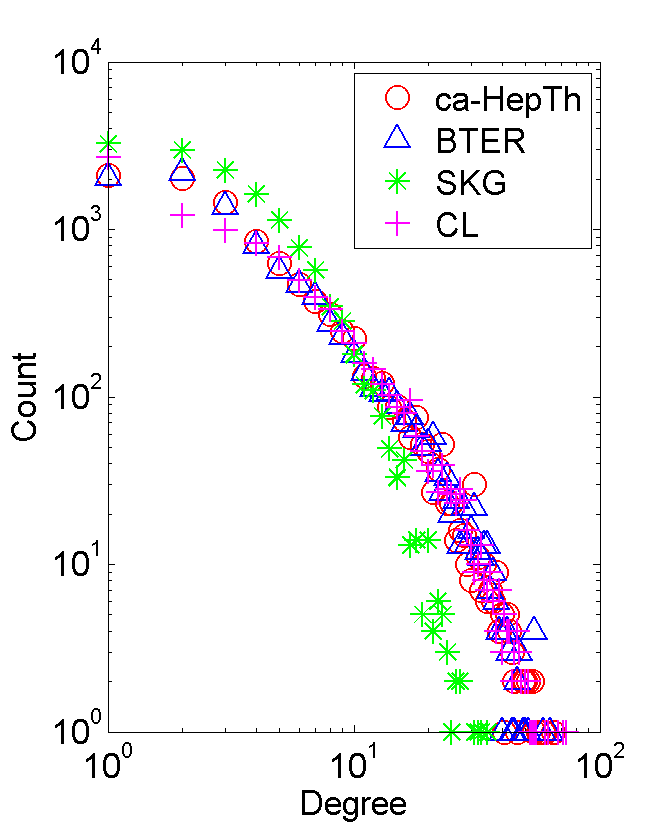}}
  \subfloat[Clustering coefficients for ca-HepTh]
  {\label{fig:ca-HepTh-cc}\includegraphics[width=0.33\textwidth]{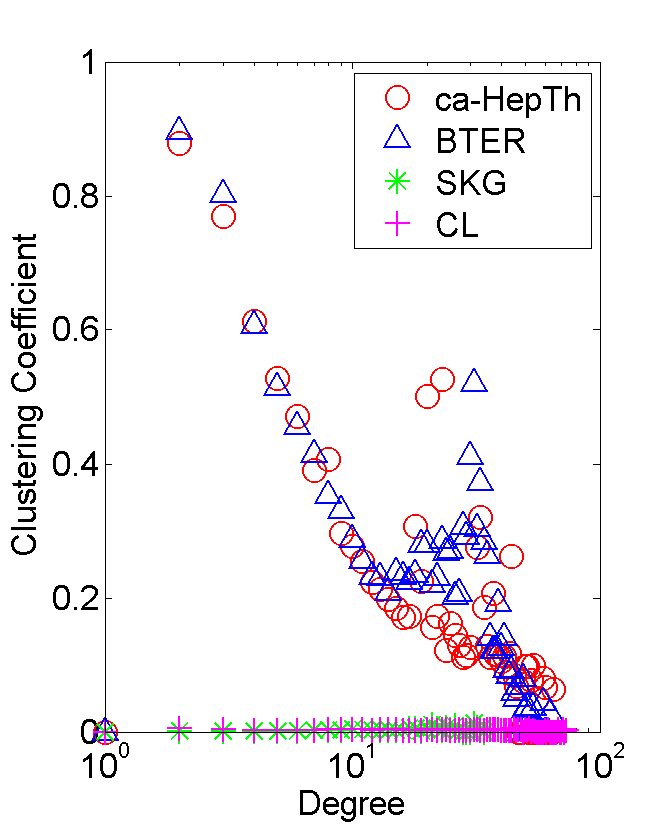}}
  \subfloat[Leading adjacency matrix eigenvalues for ca-HepTh]
  {\label{fig:ca-HepTh-eigs}\includegraphics[width=0.33\textwidth]{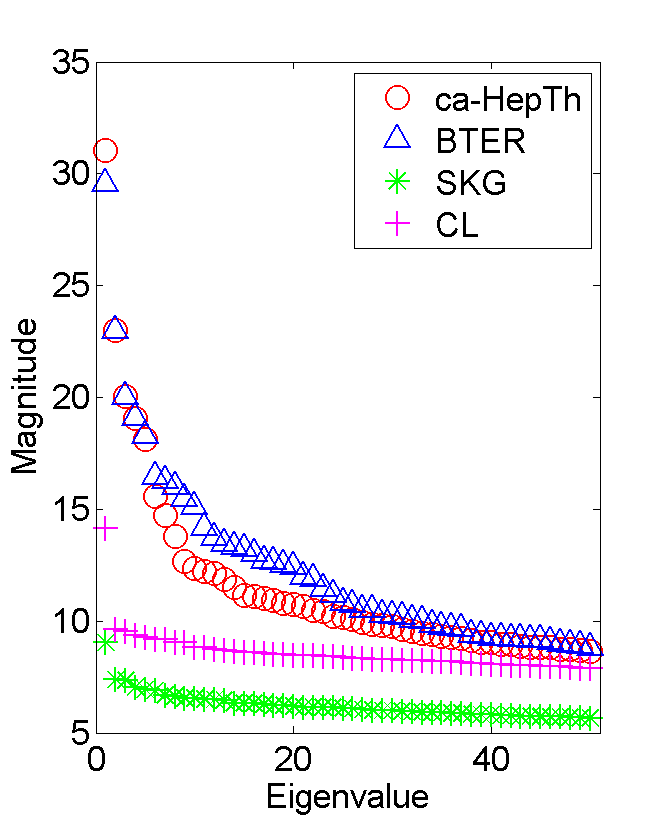}}\\
  \subfloat[Degree distribution for soc-Epinions1]
  {\label{fig:soc-Epinions1-dd}\includegraphics[width=0.33\textwidth]{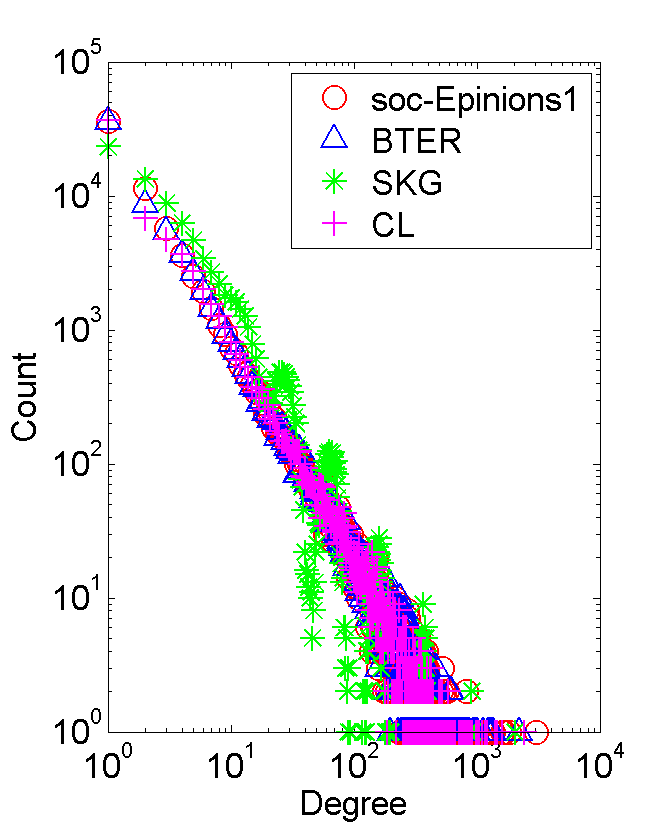}}
  \subfloat[Clustering coefficients for soc-Epinions1]
  {\label{fig:soc-Epinions1-cc}\includegraphics[width=0.33\textwidth]{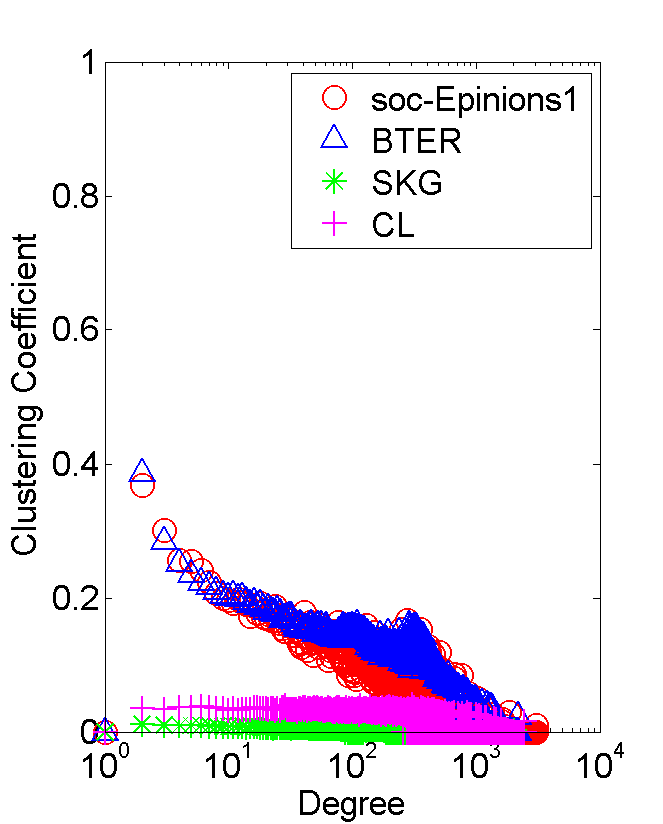}}
  \subfloat[Leading adjacency matrix eigenvalues for soc-Epinions1]
  {\label{fig:soc-Epinions1-eigs}\includegraphics[width=0.33\textwidth]{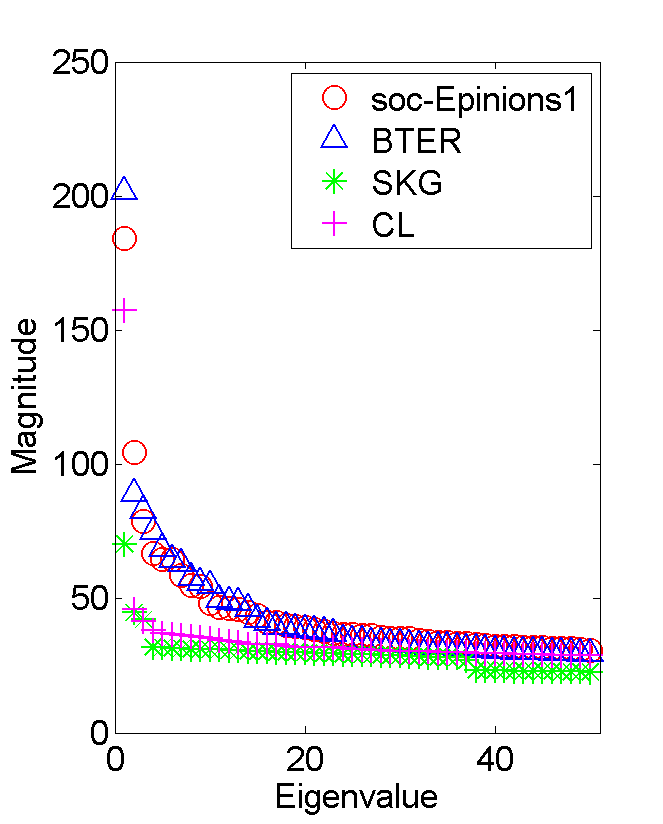}}
  \caption{Comparison of Chung-Lu (CL), SKG, and BTER on small graphs.}
  \label{fig:smalldata}
\end{figure}

\paragraph{Degree distribution}
The degree distributions for the original graphs and the models are
shown in \FigsTwo{ca-HepTh-dd}{soc-Epinions1-dd}.
SKG is known to have oscillations in the degree distribution
\cite{JACM-SKG,SePiKo11}, and these oscillations are easily visible  in
\Fig{soc-Epinions1-dd}. The oscillations are correctable with appropriate
addition of noise \cite{JACM-SKG,SePiKo11} (not shown), but even then it tends to
overestimate the low degree nodes and miss the highest degree nodes.
In contrast, both Chung-Lu and BTER closely match the real data.

\paragraph{Clustering coefficients}
The clustering coefficients per degree for the original graphs and the models are
shown in \FigsTwo{ca-HepTh-dd}{soc-Epinions1-cc}.
The SKG graph model has no inherent mechanism for closing triangles
and creating a community structure. Though a few triangles may close at
random, they are insufficient for the SKG-generated graph to match the
clustering coefficients in the real data. 
The situation for Chung-Lu is similar to that for SKG --- there is no reason
for wedges to close.
BTER, on the other hand, provides a much closer match to the real data.

\paragraph{Eigenvalues of adjacency matrix}
We show the top 50 leading eigenvalues (in magnitude) of the
adjacency matrix in \FigsTwo{ca-HepTh-eigs}{soc-Epinions1-eigs}.  BTER
provides a much closer match to the real data ---
especially the first few eigenvalues. 
Under certain circumstances,
matching the degree distribution should produce a match in eigenvalues
\cite{MiPa02}. However, based on the observation that the eigenvalues of a graph with community structure are larger than those of the Chung-Lu graph with the same degree distribution, we conjecture that graphs with community structure
require that triangle structure also be matched to obtain a good fit for the
eigenvalues.

\subsection{Large data}

We demonstrate that BTER is able to fit large-scale real-world
data. We do not compare to SKG because it is not possible to fit the
parameters for such large graphs.  We do not compare to the Chung-Lu model
because we can easily explain the performance: its match in terms of
the degree distribution is nearly identical to that of BTER, and its
clustering coefficients are close to zero, for the small data.
The data sets are described in \Tab{data}. We treat all edges as
undirected and remove any duplicate edges and loops. 
We obtained real-world graphs from the Laboratory for Web Algorithms
\cite{LWA}, which compressed the graphs using
LLP and WebGraph \cite{BoVi04,BoRoSaVi11}.
Briefly, the networks are described as follows.
\begin{compactitem}
\item amazon-2008 \cite{BoVi04,BoRoSaVi11}: A graph describing similarity among books as
  reported by the Amazon store.
\item ljournal-2008 \cite{ChKuLaMi09,BoVi04,BoRoSaVi11}: Nodes represent users on
  LiveJournal. Node $x$ connects to node $y$ if $x$ registered $y$ as
  a friend.
\item hollywood-2011 \cite{BoVi04,BoRoSaVi11}: This is a graph of actors. Two
  actors are joined by an edge whenever they appear in a movie together.
\item twitter-2010 \cite{KwLePaMo10,BoVi04,BoRoSaVi11}: Nodes are Twitter users, and
  node $x$ links to node $y$ if $y$ follows $x$. 
\item uk-union-2006-06-2007-05 (shortened to
  uk-union)
  \cite{BoSaVi08,BoVi04,BoRoSaVi11}: Links between web pages on the .uk domain. We
  ignore the time labeling on the links.
\end{compactitem}
To the best of our knowledge, uk-union is the largest publicly available graph.

\begin{table}[htbp]
 \caption{Network characteristics of original and BTER-generated graphs.}
  \label{tab:graphs}
  \centering\small
\subfloat[Large data set properties.]{\label{tab:data}
\begin{tabular}{|l|r|r|r|r|r|} \hline
    \textbf{Graph}  & \ColTitle{$|V|$} & \ColTitle{$|E|$} & 
    \ColTitle{$\dmax$} & \ColTitle{$\davg$} & \ColTitle{GCC} \\ \hline
    amazon-2008 &   1M &     4M &   1,077 &  10 & 0.260 \\ \hline
  ljournal-2008 &   5M &    50M &  19,432 &  18 & 0.124 \\ \hline
 hollywood-2011 &   2M &   114M &  13,107 & 115 & 0.175 \\ \hline
   twitter-2010 &  40M & 1,202M & 2,997,487 &  60 & 0.001 \\ \hline
       uk-union & 122M & 4,659M & 6,366,528 &  76 & 0.007 \\ \hline
  \end{tabular}
}

\subfloat[Properties of BTER-generated graphs, including generation and edge deduplication time.]{
\label{tab:bter}\small
\begin{tabular}{|l|r|r|r|r|r|r|r|} \hline
    \textbf{Graph}  & \ColTitle{$|V|$} & \ColTitle{$|E|$} & \ColTitle{$\dmax$} & 
    \ColTitle{$\davg$} & \ColTitle{GCC} & \ColTitle{Gen.} & \ColTitle{Dedup.}\\ \hline
    amazon-2008 &   1M &     4M &   1,052 &  10 & 0.253 &  2.27s &  9.25s \\ \hline
  ljournal-2008 &   5M &    49M &  18,510 &  19 & 0.127 & 33.81s & 126.40s \\ \hline
 hollywood-2011 &   2M &   114M &  11,676 & 115 & 0.180 & 88.54s & 362.25s \\ \hline
   twitter-2010 &  38M & 1,135M & 1,635,823 &  59 & 0.004 & \multicolumn{2}{c|}{230s} \\ \hline
       uk-union & 120M & 4,405M & 1,497,950 &  73 & 0.111 & \multicolumn{2}{c|}{1350s} \\ \hline
  \end{tabular}
}
 \end{table}

The smaller graphs (amazon-2008, ljournal-2008, hollywood-2011) are
those with up to roughly 100M edges. These can be easily processed
using MATLAB on an SGI Altix UV 10 with 32 cores (4 Xeon 8-core 2.0GHz
processors) and 512 GB DDR3 memory. None of the parallel capabilities
of MATLAB are enabled for these studies. To give a sense of the memory
requirements, storing the hollywood-2011 graph as a sparse matrix in
MATLAB requires 3.4GB of storage.
The larger graphs (twitter-2010, uk-union) each have over 1B edges.
These are processed on a
Hadoop cluster with 32 compute nodes. Each compute node has an
Intel i7 930 CPU at 2.8GHz (4 physical cores, HyperThreading enabled),
12 GB of memory, and 4 2TB SATA disks. All experiments were run using
Apache Hadoop version 0.20.203.0\footnote{\url{http://hadoop.apache.org/}}.
The results in \Tab{bter} were obtained using 132 map and 32 reduce tasks.

The inputs to BTER are the degree distribution and clustering coefficients by degree. 
(We used a blowup of $\beta=1$ for the experiments reported here.)
Computing the degree distribution is straightforward.
However, for the clustering coefficients calculations, we used the sampling
approach as implemented in \cite{KoPiPlSe13} with 2000 samples per degree,
so the expected error is $\epsilon=0.05$ at a confidence level of
99.9\%. Sampling was not required for the smaller graphs, but we
used it in all experiments for consistency.

\paragraph{BTER Timing}
\Tab{bter} shows the details and timings for the graphs produced by
BTER. Observe the close match in the characteristics of the graphs in
terms of number of nodes, number of edges, maximum degree, average
degree, and global clustering coefficient.  For the smaller graphs, we
are able to separate the edge generation and deduplication time. The
generation time is not strictly proportional to the number of desired
edges because we have to generate extra edges for Phase 1 to account
for possible duplicates (see \Sec{phase1}). The parallelism of Hadoop
yields a large advantage in terms of time. The twitter-2010 graph has
10 times more edges than hollywood-2011, but it takes less than half
the time to do the computation on the 32-node Hadoop cluster.

\paragraph{Degree Distribution}
\Fig{degdist-realdata} illustrates the match between the real data and
the BTER graph. BTER cannot easily match discontinuities in the degree
distribution because of the randomness in creating edges. The
issue is that nodes generally do not get \emph{exactly} the desired
degree; in the realization of the graph, observed degrees may deviate
by one or two from the expected degree. For a smooth degree
distribution, neighboring degrees cancel the effect of one on another.
For discontinuous distributions, the BTER degree distribution is a
``smoothed'' version. This is evident, for instance, in the
amazon-2008 data where we can see a smoothing effect on the sharp
discontinuity near degree 10. 

\begin{figure}[htbp]
  \centering
  \subfloat[amazon-2008]{\includegraphics[width=0.33\textwidth]{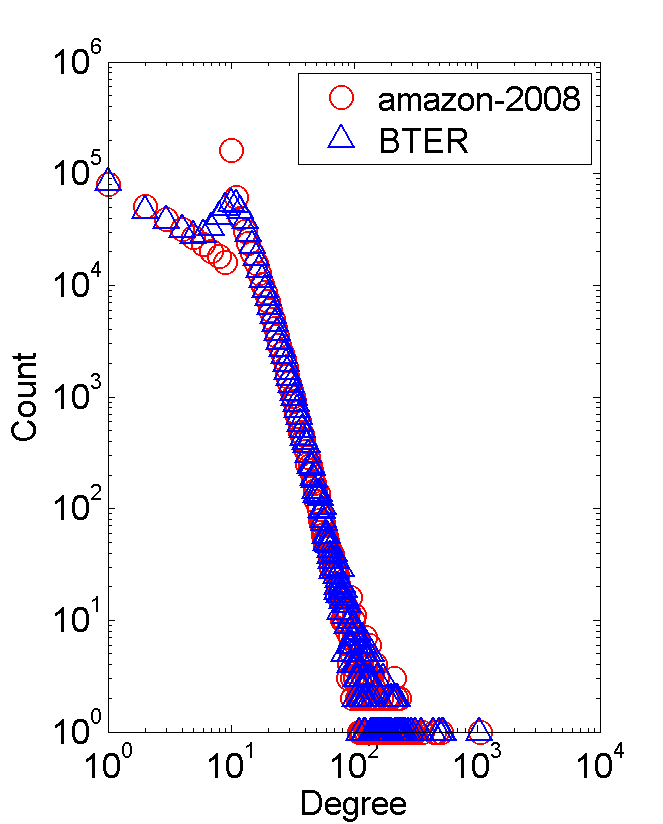}}
  \subfloat[ljournal-2008]{\includegraphics[width=0.33\textwidth]{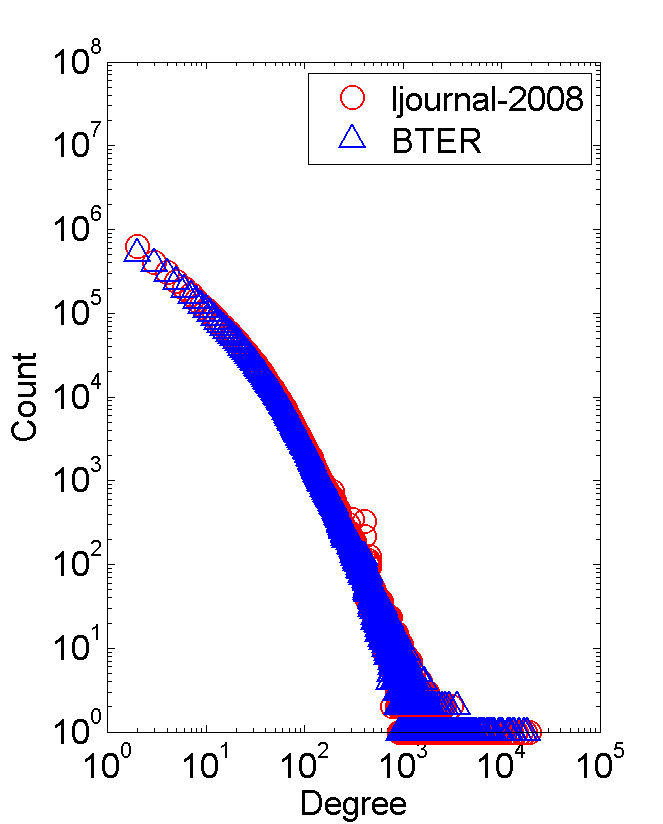}}
  \subfloat[hollywood-2011]{\includegraphics[width=0.33\textwidth]{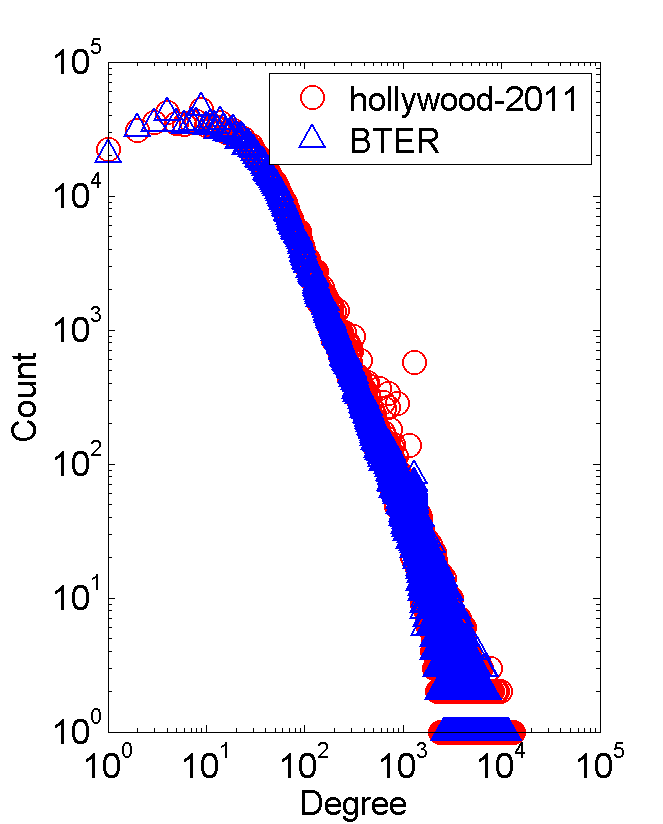}}\\
  \subfloat[twitter-2010]{\includegraphics[width=0.33\textwidth]{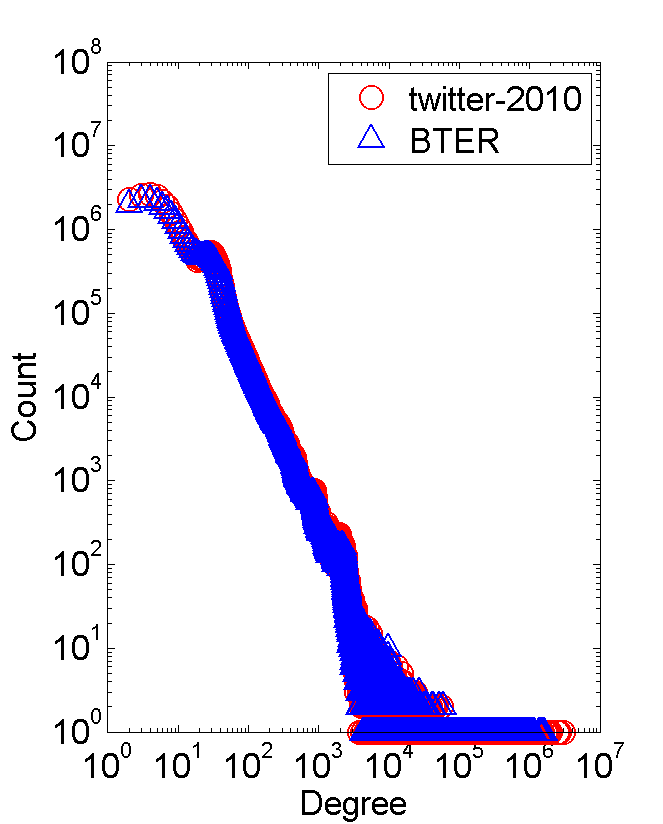}}
  \subfloat[uk-union]{\includegraphics[width=0.33\textwidth]{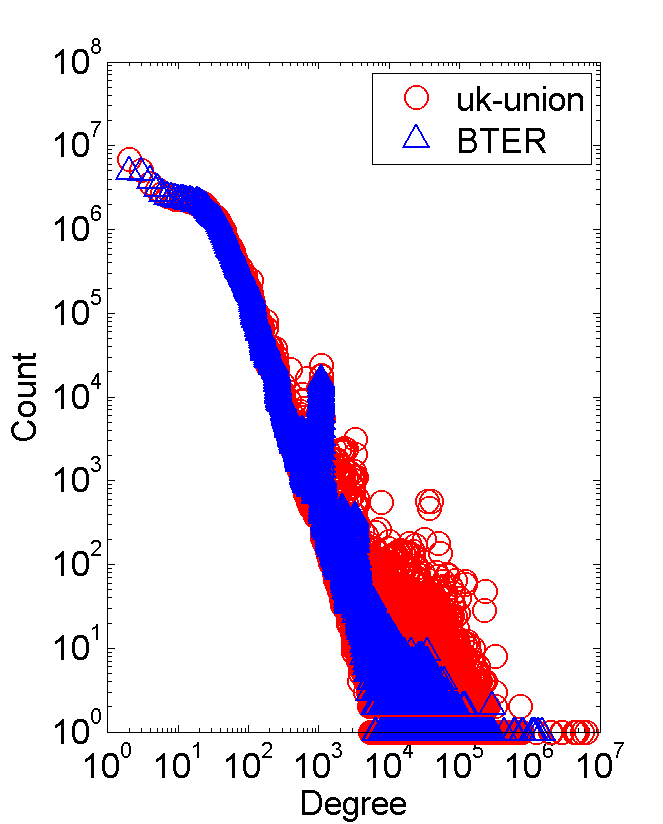}}
  \caption{Degree distributions of original and BTER-generated graphs.}
  \label{fig:degdist-realdata}
\end{figure}

\paragraph{Clustering Coefficients}
BTER's strength is its ability to match clustering coefficients and
therefore create community structure. Most degree distributions are heavy-tailed
and have a relatively consistent structure. The same is not
true for clustering coefficients. Different profiles can potentially
lead to graphs with fundamentally different
structures. \Fig{cc-realdata} shows the clustering coefficients of the
real data and the BTER-generated graphs. There is a very close match.

\begin{figure}[htbp]
  \centering
  \subfloat[amazon-2008]{\includegraphics[width=0.33\textwidth]{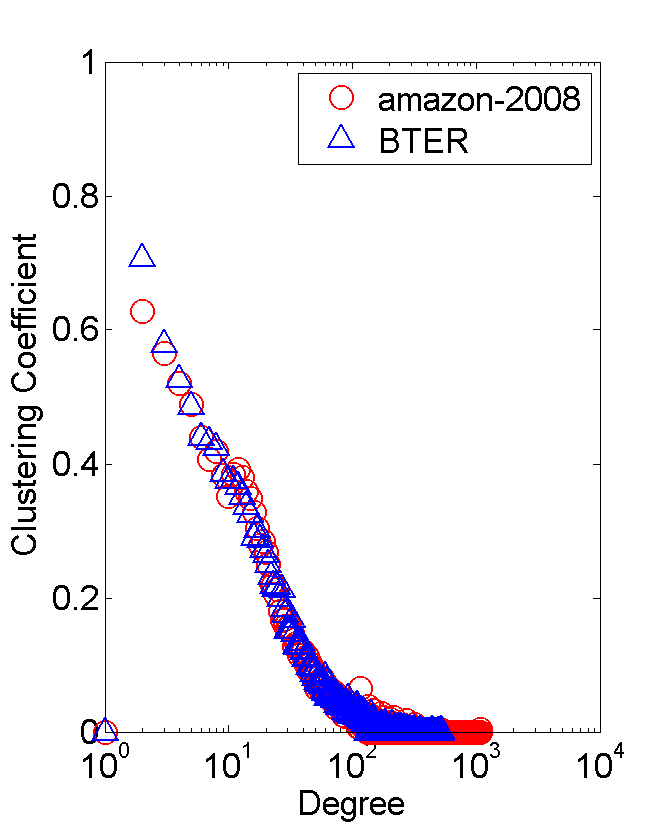}}
  \subfloat[ljournal-2008]{\includegraphics[width=0.33\textwidth]{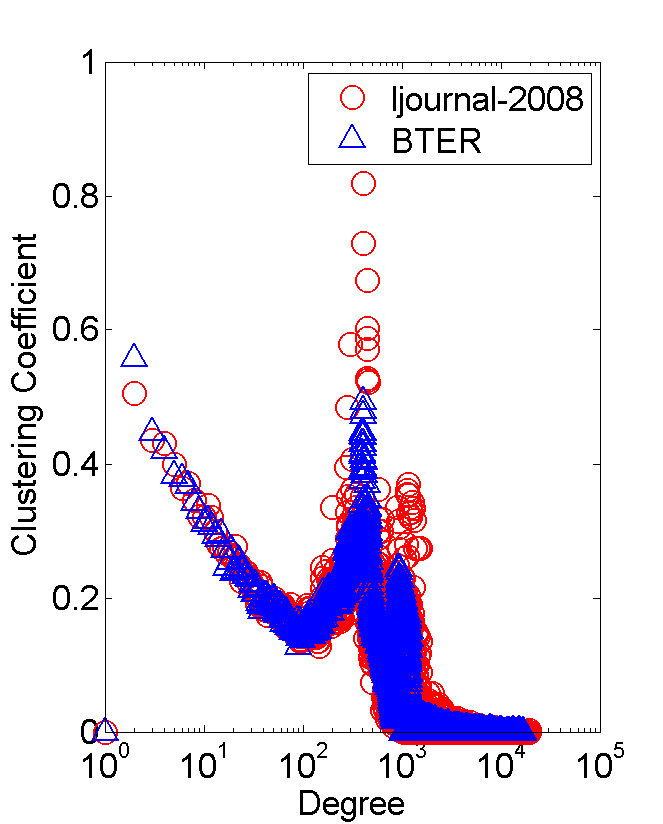}}
  \subfloat[hollywood-2011]{\includegraphics[width=0.33\textwidth]{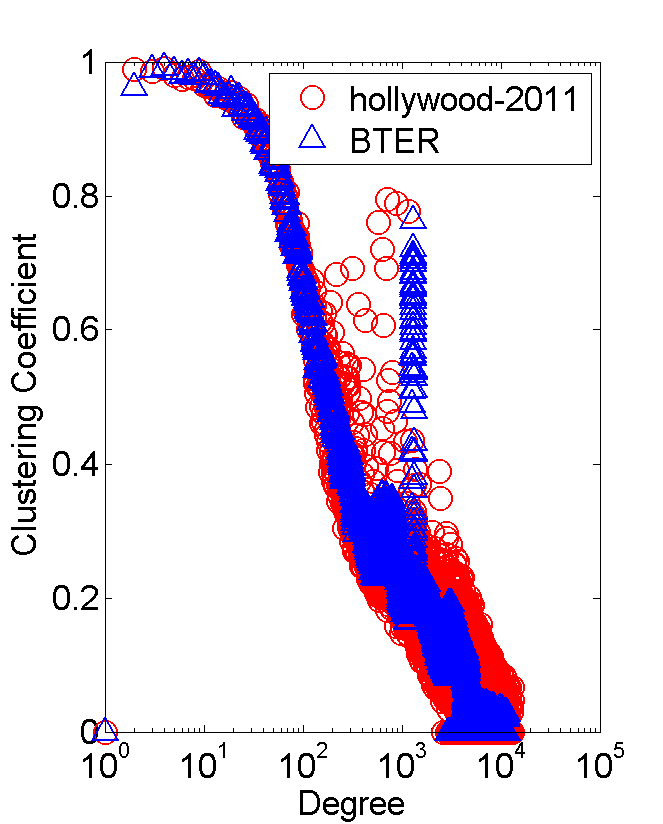}}\\
  \subfloat[twitter-2010]{\includegraphics[width=0.33\textwidth]{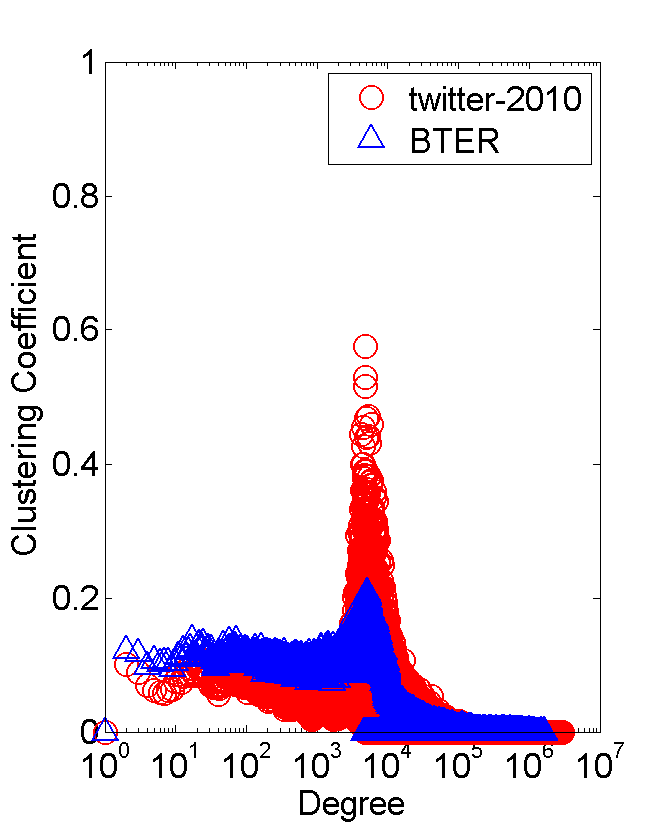}}
  \subfloat[uk-union]{\includegraphics[width=0.33\textwidth]{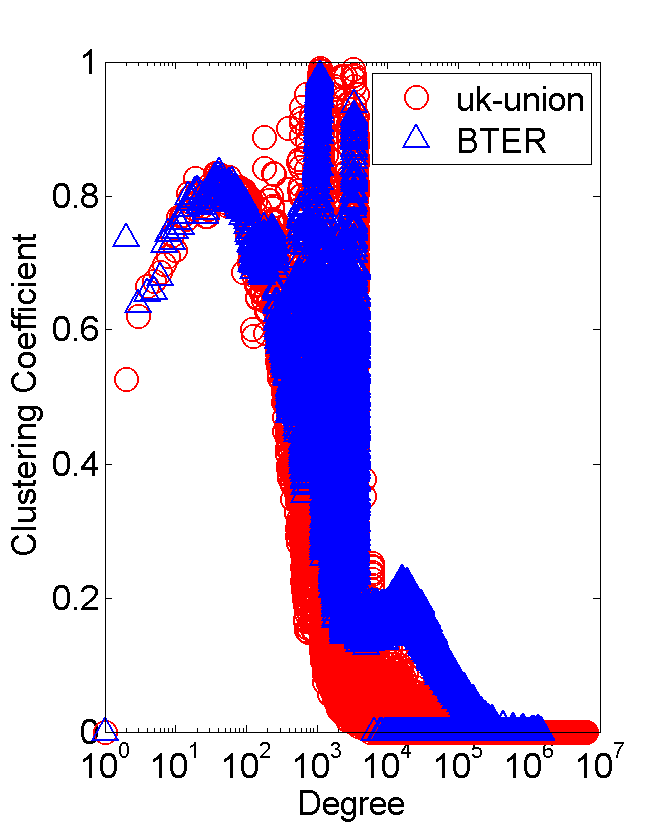}}
  \caption{Clustering coefficients of original and BTER-generated graphs.}
  \label{fig:cc-realdata}
\end{figure}

\section{Proposed Benchmark Parameters and Scalability} \label{sec:benchmark}
Thus far we have considered how BTER can be used to match real-world data.  For
benchmarking purposes, where there is no specific graph to be matched, ``ideal''
profiles for degree distribution and clustering coefficient by degree are
required. In this section, we propose some possibilities, noting that these are
tunable for various testing scenarios, i.e., specifying an average degree, a
maximum clustering coefficient, etc.
Generating an artificial degree distribution can be problematic. For instance,
using a straight power law distribution can lead to impossible  situations, like
choosing a degree greater than the number of nodes. This necessitates a discrete
distribution, which engenders its own problems.
After proposing methods for working around these problems, 
we use the proposed benchmark in a scalability study for the MapReduce implementation of BTER.

\subsection{Idealized Degree Distribution}
\label{sec:idealdeg}

It has been hypothesized that degree distribution of real-world
networks follow a power law (PL) degree distribution, i.e., 
\begin{displaymath}
  n_d \propto d^{-\gamma},
\end{displaymath}
for some parameter $\gamma$ \cite{BaAl99}. However, our observation is
that power law distributions are difficult to use as a model --- a
point that is discussed in more detail below.  It has been suggested
that power laws are not necessarily the best descriptors for real-world
networks \cite{SaGaRoZh11,BiFaKo01}. Finally, proving (in a
statistical sense) that a single observed degree distribution is power
law is difficult \cite{ClShNe09}.

For benchmarking purposes, our goal is to specify an ideal average degree,
$\davgideal$, and an absolute bound on maximum degree, $\dmaxideal$. Let $f(d)$ define the desired
proportionality of degree $d$, e.g., $f(d) = d^{-\gamma}$ for a power law distribution. We then create a \emph{discrete}
distribution on $d=1,\dots,\dmaxideal$ as
\begin{displaymath}
  \Prob{D = d} = \frac{f(d)}{\sum_{d'=1}^{\dmaxideal} f(d')}.
\end{displaymath}
Ideally, the average degree is equal to $\davgideal$ and the probability
of having degree $\dmaxideal$ is sufficiently small, i.e.,
\begin{displaymath}
  \davgideal = \sum_{d=1}^{\dmaxideal} d \cdot f(d)
  \qtext{and}
  \Prob{D = \dmaxideal} < \epsilon_{\rm tol},
\end{displaymath}
where $\epsilon_{\rm tol}$ is small enough such that $n \cdot
\epsilon_{\rm tol} \ll 1$ (where $n$ is the number of nodes). 
For the power law distribution, it can be difficult to find a value for
$\gamma$ that yields a high enough average degree and a low enough
probability of choosing $\dmaxideal$.  Hence, we propose instead a
generalized log-normal (GLN) distribution, i.e.,
\begin{displaymath}
  n_d \propto \exp\left[-\left(\frac{\log d}{\alpha}\right)^{\delta}\right],
\end{displaymath}
for some parameters $\alpha$ and $\delta$. 
The appendices includes example degree distributions that vary $\alpha$ and $\delta$.
The shape of the distribution is typical of the real-world graphs shown in \Sec{realdata}. 

We consider two scenarios, both with $n=10^7$ nodes. We do a parameter search on $\alpha$ and $\delta$ (\texttt{fminsearch} in MATLAB) to locate the optimal parameters. 
A function \texttt{degdist\_param\_search}
that finds the optimal parameters for either discrete GLN
or discrete PL for user-specified values of $\davg$ and $\dmax$ is
included in the reference code to be released at a future date.

\paragraph{Scenario 1 for Degree Distribution Fitting} 
In the first scenario, the targets are $\davgideal = 16$ and $\dmaxideal = 10^{6}$.
For discrete PL, the optimal parameter is $\gamma=1.911$ with $\davg=16$ and $\Prob{D=\dmaxideal} = 1.97 \times 10^{-12}$.
For  discrete GLN, the optimal parameters are $\alpha=1.988$ and $\delta=2.079$ with $\davg=16$ and $\Prob{D=\dmaxideal}=4.14 \times 10^{-26}$.
Realizations of the two distributions are pictured in \Fig{degdistcompare1}.
For this scenario, both degree distributions are reasonable in that
there is no sharp dropoff as we get close to the maximum allowable
degree, $\dmaxideal$.

\paragraph{Scenario 2 for Degree Distribution Fitting} 
In the second scenario, the targets are $\davgideal = 64$ and $\dmaxideal = 10^{5}$.
For PL, the optimal parameter is $\gamma=1.668$ with $\davg=64$ but
$\Prob{D=\dmaxideal} = 2.16 \times 10^{-9}$ (fairly large).
For discrete GLN, the optimal parameters are $\alpha=2.171$ and $\delta=1.877$ with $\davg=64$ and $\Prob{D=\dmaxideal}=8.35 \times 10^{-12}$.
Realizations of the two distributions are pictured in \Fig{degdistcompare2}.
In this scenario, the problem with power law becomes apparent --- near $\dmaxideal$, there are still many degrees with \emph{multiple} nodes so that the cutoff is extremely abrupt. In comparison, discrete GLN fades more naturally to the desired maximum degree.

\begin{figure}[thbp]
  \centering
  \subfloat[Scenario 1: $\davgideal = 16$ and $\dmaxideal = 10^{6}$]
  {\label{fig:degdistcompare1}\includegraphics[width=0.33\textwidth]{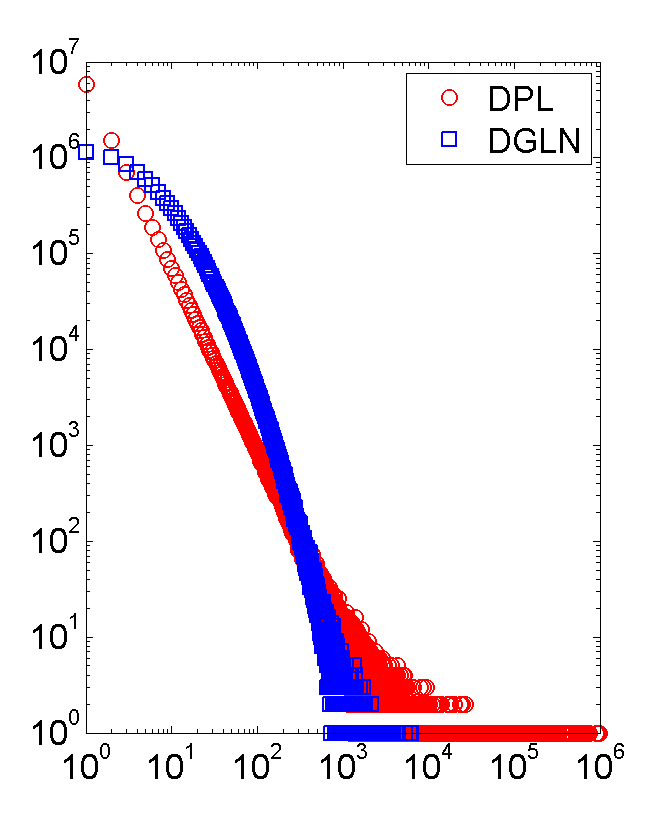}}
  ~~~~~
  \subfloat[Scenario 2: $\davgideal = 64$ and $\dmaxideal = 10^{5}$]
  {\label{fig:degdistcompare2}\includegraphics[width=0.33\textwidth]{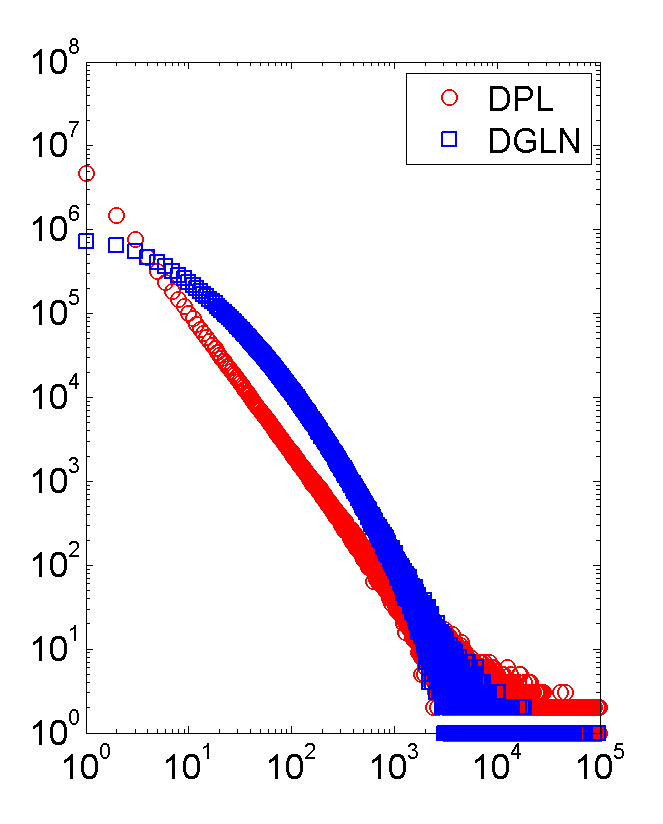}}
  \caption{Example degree distributions from discrete power law (DPL)
    and discrete generalized log normal (DGLN) for $n=10^{7}$ nodes.}
  \label{fig:degdistcompare}
\end{figure}

\subsection{Idealized Clustering Coefficients}
\label{sec:idealcc}

As there is no definitive structure to clustering coefficients, we
propose a simple parameterized curve that has some
similarity to real data observations.

Let $\set{n_d}$ define the specified degree distribution and $\dmax$
be the maximum degree such that $n_d > 0$. We define  $\bar c_d$, the mean value for
$c_d$, as 
\begin{displaymath}
  \bar c_d = c_{\rm max} \exp(-(d-1)\cdot {\xi}) \qtext{for} d \geq 2,
\end{displaymath}
where $c_{\rm max}$ and $\xi$ are parameters. If $c_{\rm max}$ is
specified, then a simple parameter search can be used to fit $\xi$ to
a target global clustering coefficient; code to fit the data is
included in the reference code.  The final values for $\set{c_d}$
are selected as
\begin{displaymath}
  c_d \sim \mathcal{N}(\bar c_d, \min\{10^{-2}, \bar c_d/2\}).
\end{displaymath}
The randomness could, of course, be omitted.

\subsection{Example Graphs}

We generate two example graphs per the scenarios below.  \Tab{ideal}
lists the network characteristics and \Fig{ideal} shows the target and
BTER-generated degree distributions and clustering coefficients.

\begin{table}[htbp]\small
  \caption{Network characteristics of BTER-generated graphs for benchmarking.}
  \label{tab:ideal}
  \centering
\begin{tabular}{|l|r|r|r|r|r|r|r|} \hline
    \textbf{Graph}  & \ColTitle{$|V|$} & \ColTitle{$|E|$} & \ColTitle{$\dmax$} & 
    \ColTitle{$\davg$} & \ColTitle{GCC} & \ColTitle{Gen.} & \ColTitle{Dedup.}\\ \hline
        Scenario 1 &   1M &    35M &  28,643 &  72 & 0.406 & 35.11s & 117.18s \\ \hline
        Scenario 2 &   1M &     8M &   2,594 &  17 & 0.104 &  5.07s & 20.66s \\ \hline
  \end{tabular}
\end{table}

\begin{figure}[htbp]
  \centering
  \subfloat[Degree distribution for Scenario 1]{\includegraphics[width=0.33\textwidth]{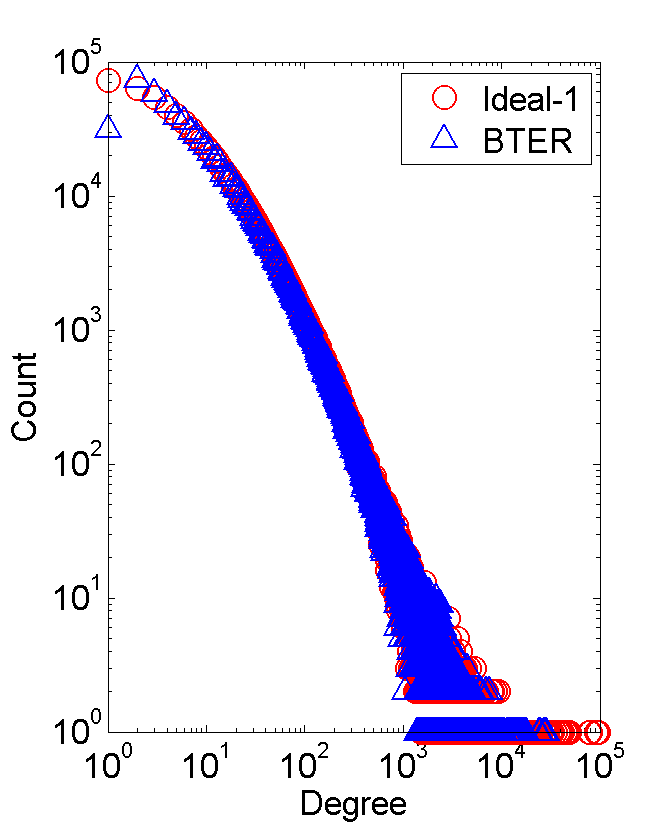}}
  \subfloat[Degree distribution for Scenario 2]{\includegraphics[width=0.33\textwidth]{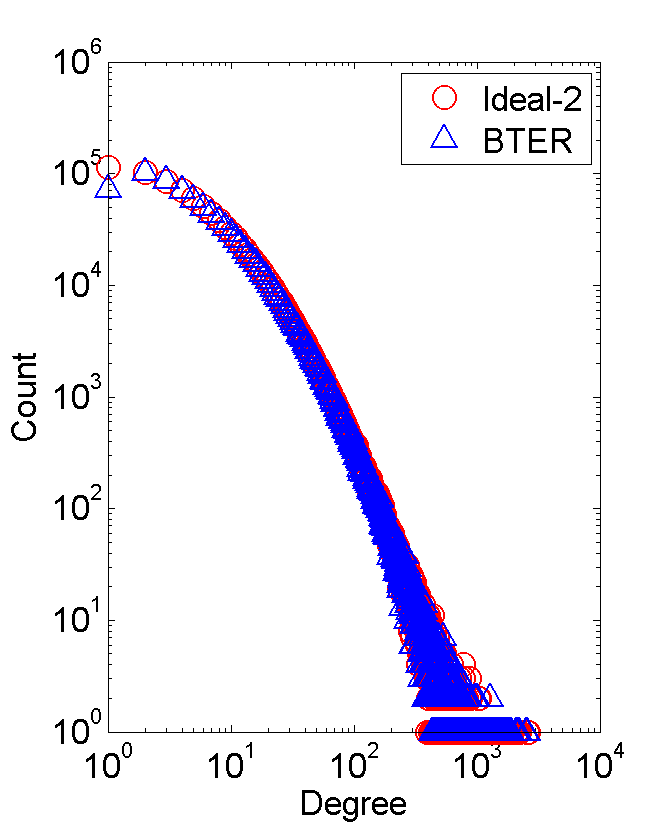}}\\
  \subfloat[Clustering coefficients for Scenario 1]{\includegraphics[width=0.33\textwidth]{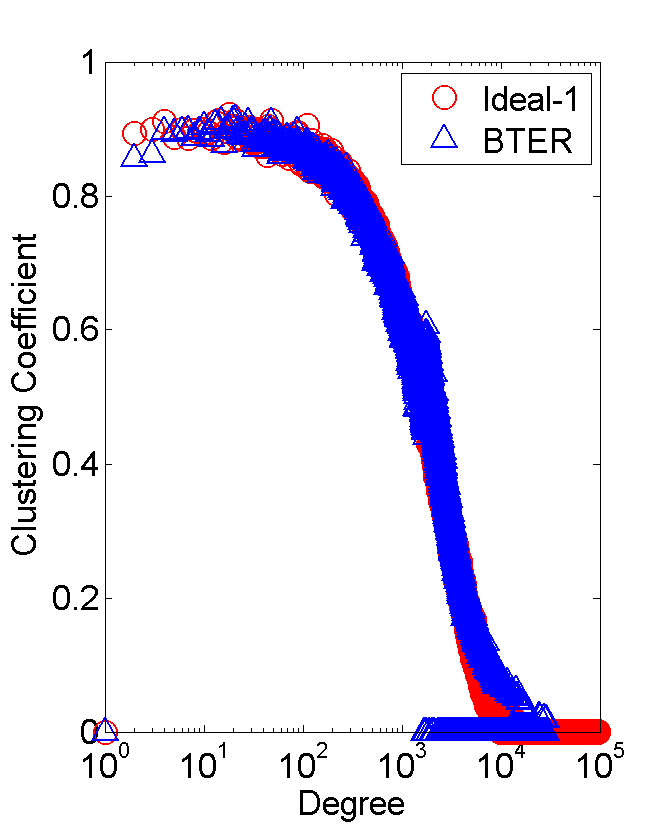}}
  \subfloat[Clustering coefficients for Scenario 2]{\includegraphics[width=0.33\textwidth]{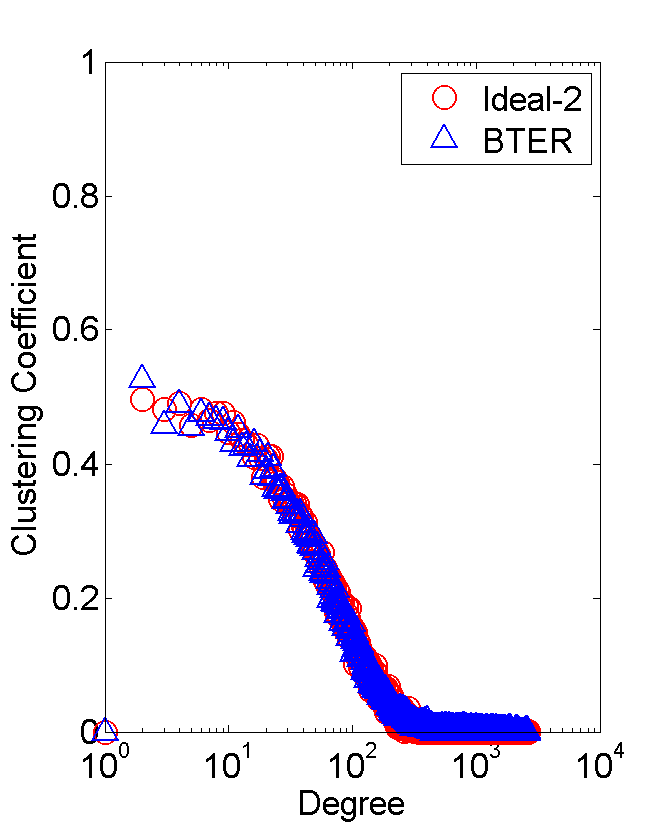}}
  \caption{Target distributions and results of BTER-generated graphs.}
  \label{fig:ideal}
\end{figure}

\paragraph{Scenario 1}
For the first set-up, we selected $\davgideal = 75$ and $\dmaxideal=100,000$ to
define the degree distribution. The parameter search selected
$\alpha=2.14$ and $\delta=1.83$. For the clustering coefficients, we
set $c_{\rm max} = 0.9$ and a target GCC of 0.15. The parameter
search selected $\xi=3.59 \times 10^{-4}$ for defining the clustering
coefficient profile.

\paragraph{Scenario 2}
For the second set-up, we selected $\davgideal = 16$ and $\dmaxideal=10,000$ to
define the degree distribution. The parameter search selected
$\alpha=1.98$ and $\delta=2.08$. For the clustering coefficients, we
set $c_{\rm max} = 0.5$ and a target GCC of 0.10. The parameter
search selected $\xi=0.01$ for defining the clustering coefficient profile.

\subsection{Scalability Test}
\label{sec:scalability}
We use the DGLN distribution to create target distributions that define a
series of graphs of different sizes but similar community structure.
From these, we generate example graphs with our Hadoop MapReduce implementation
of BTER and demonstrate scalability of our code.

\Tab{scalability} describes characteristics of the test
graphs.  We selected $\davgideal = 32$ for all graphs and maintained $\dmaxideal$
proportional to $\sqrt{|V|}$, a relation that we often observe in social
network graphs.
The last two columns of the table list the number of generated
(phase 1 and 2) and unique edges measured from full graphs realized by BTER.

Target degree distributions were created using the formulas and functions
presented in \Sec{idealdeg}.
Optimal parameters $\alpha$ and $\delta$ were
computed using MATLAB function \texttt{degdist\_param\_search} (part of the
reference code).  The function takes $\davgideal$ and $\dmaxideal$ as inputs,
plus the $\Prob{D=\dmaxideal}$ which we set to $0.001 / |V|$.
Then we computed a discrete GLN from $\alpha$, $\delta$, and the number
of nodes $|V|$ listed in the second column of \Tab{scalability}.
We created clustering coefficient distributions as in \Sec{idealcc},
setting $c_{\rm max} = 0.5$ and $\mbox{GCC} = 0.15$ for all graphs.

\begin{table}[htbp]\small
  \caption{Graph characteristics for testing BTER scalability.}
  \label{tab:scalability}
  \centering
  \begin{tabular}{|c|r|r|r|r|r|r|r|} \hline
    \textbf{Graph}  & \ColTitle{$|V|$} & \ColTitle{$\dmaxideal$} & 
      \ColTitle{$\davgideal$} & \ColTitle{gen $|E|$} & \ColTitle{unique $|E|$}
    \\ \hline
  1 &   1M &  50k & 32 &    25M &    16M  \\
  2 &   2M &  70k & 32 &    50M &    32M  \\
  3 &   4M & 100k & 32 &   101M &    64M  \\
  4 &   8M & 140k & 32 &   201M &   128M  \\
  5 &  16M & 200k & 32 &   403M &   256M  \\
  6 &  32M & 280k & 32 &   808M &   512M  \\
  7 &  64M & 400k & 32 & 1,614M & 1,024M  \\
  8 & 128M & 570k & 32 & 3,233M & 2,047M  \\
  9 & 256M & 800k & 32 & 6,467M & 4,096M  \\
    \hline
  \end{tabular}
\end{table}

The MapReduce implementation of BTER was executed on a 32-node cluster
running Apache Hadoop version 0.20.203.0.
Each compute node contains a quad-core processor and four
hard drives configured to run independently (no RAID striping).  Hadoop
was configured for a maximum of 128 simultaneous map tasks, effectively
pairing each core with a hard drive for maximum I/O throughput.  A MapReduce
job can therefore launch up to 128 mappers in parallel.  Hadoop was also
configured for a maximum of 128 simultaneous reduce tasks.  BTER uses the
reduce phase to remove duplicate edges, so we expect this to compete for
machine resources with
BTER map tasks that generate phase 1 and phase 2 edges.

\Fig{scalability1} renders the data in \Tab{scalability},
showing how the computational work varies across test graphs.
To assess scalability in the weak sense, we set the workload of each map task
in BTER to be 1 million edges; hence, Graph \#1 executes with 1 map task and
Graph \#9 executes with 256.  The number of reduce tasks was set to the number
of map tasks, up to the configured Hadoop limit of 128.

\Fig{scalability2} plots BTER execution time (wall clock time) as a
function of the number of vertices requested.  Perfect scalability would
appear as a horizontal line.  We observe excellent scalability up to 32 million
vertices (32 map tasks), and then a significant dropoff in parallel performance.
However, map task scalability is excellent through 128 tasks, the full
capacity of the cluster.  Hence, BTER edge generation in the map phase is
fully scalable to the available hardware.  Removal of duplicate edges in the
reduce phase is less scalable.  Here the edges must be sorted on the map task
node, shuffled across the network to reduce task nodes, and merged at
the reducers.
Close examination of Hadoop log files and network traffic suggest that
the reduce phase suffers from bandwidth limitations during the shuffle
and from spillover to disk during the merge.  We conjecture that deduplication
scalability could be improved on a cluster with larger network bandwidth
and more physical memory per compute node.

\begin{figure}[thbp]
  \centering
  \subfloat[Test graph sizes]
  {\label{fig:scalability1}
   \includegraphics[width=0.33\textwidth]{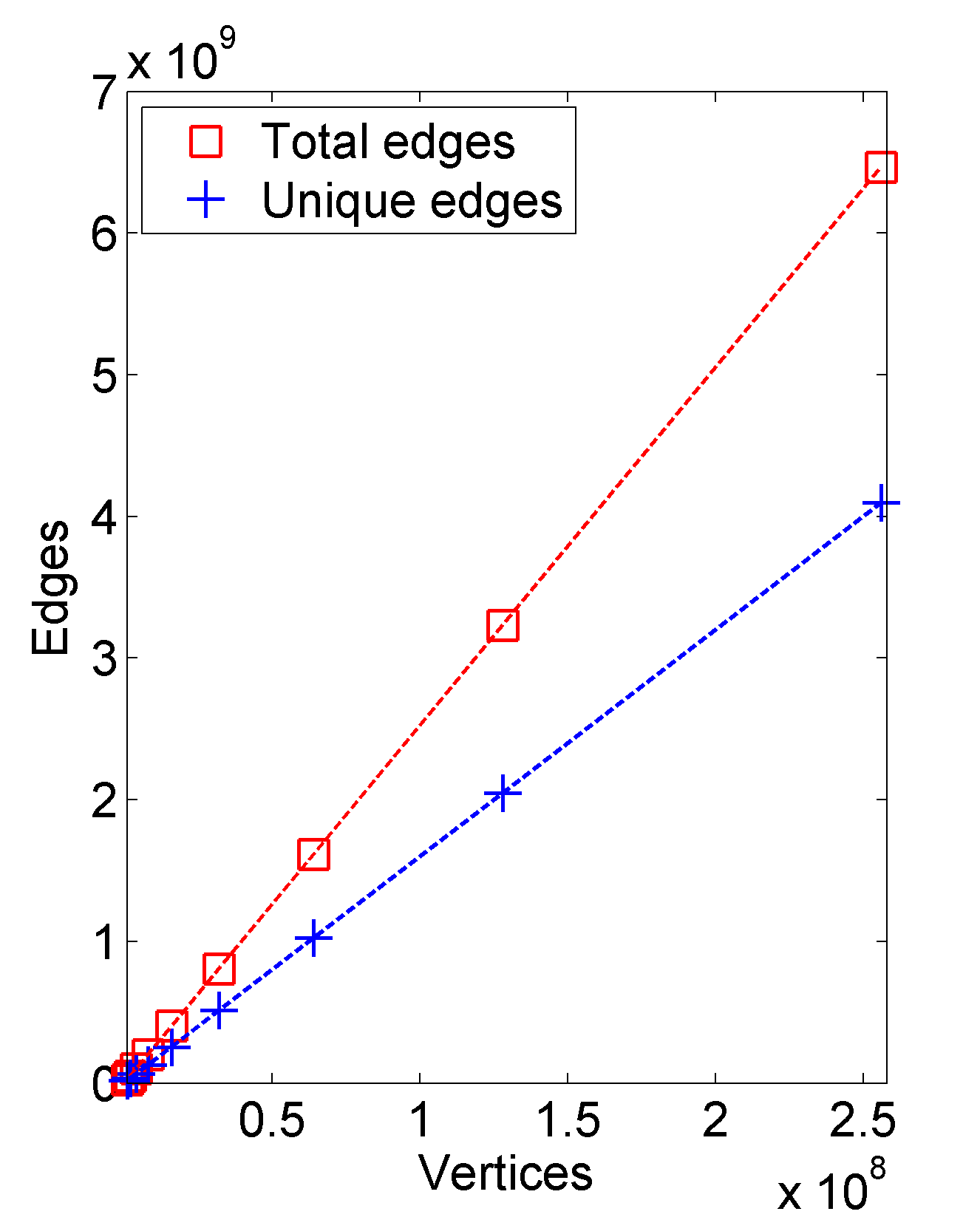} }
  ~~~~~
  \subfloat[BTER execution time in seconds]
  {\label{fig:scalability2}
   \includegraphics[width=0.33\textwidth]{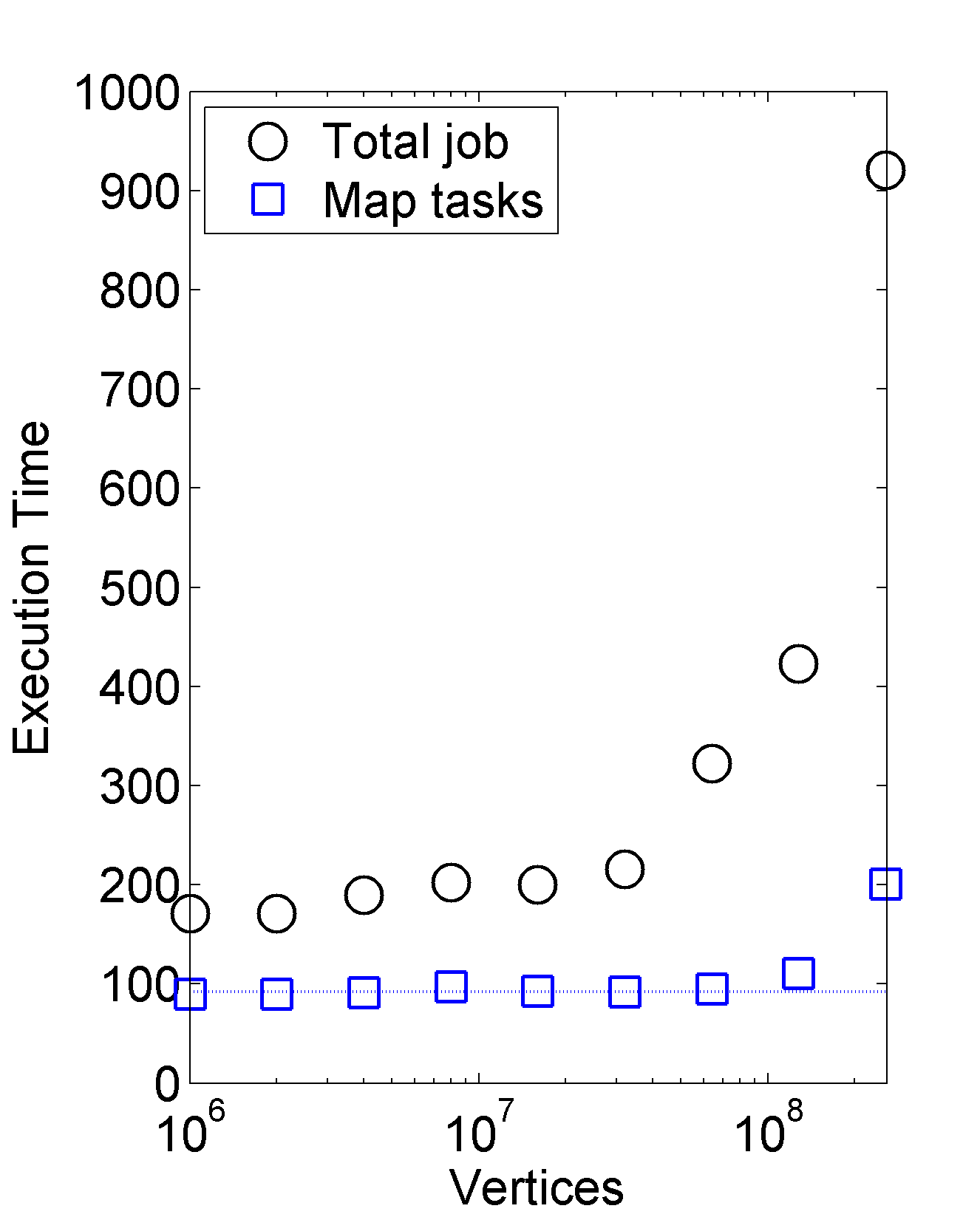} }

  \caption{Scalability of BTER MapReduce implementation.}
  \label{fig:scalability}
\end{figure}

We also show an example of the generated graph in \Fig{scalability-example}.
This is for Graph \#9 with 256M vertices and 4B edges. We give log-binned results.
The BTER model is fairly close to the intended distributions for both
the degree distribution and the clustering coefficient by degree.
\begin{figure}[thbp]
  \centering
  \subfloat[Degree distribution (log-binned)]
  {\label{fig:scalability-dd}
   \includegraphics[width=0.33\textwidth]{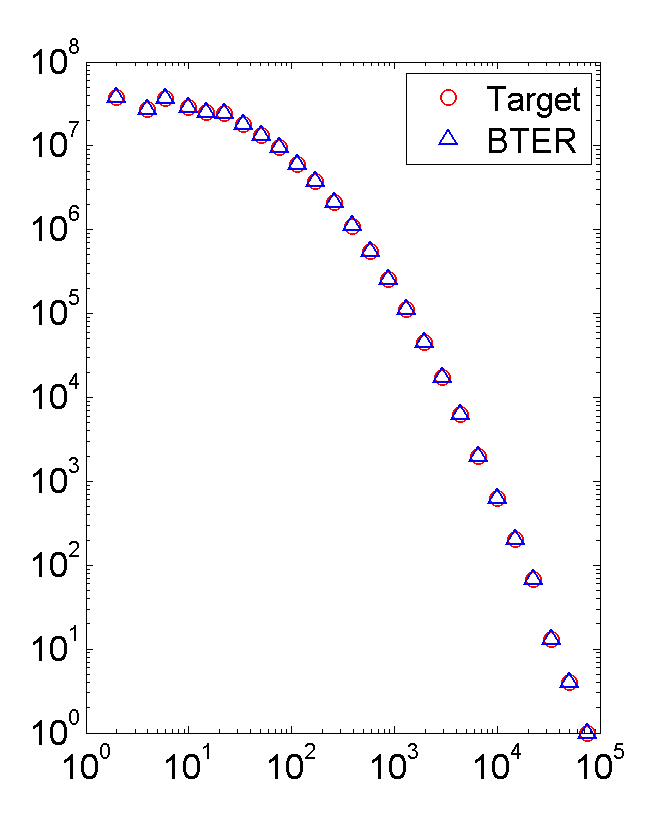}}
  ~~~~~
  \subfloat[Clustering coefficients (log-binned)]
  {\label{fig:scalability-cc}
   \includegraphics[width=0.33\textwidth]{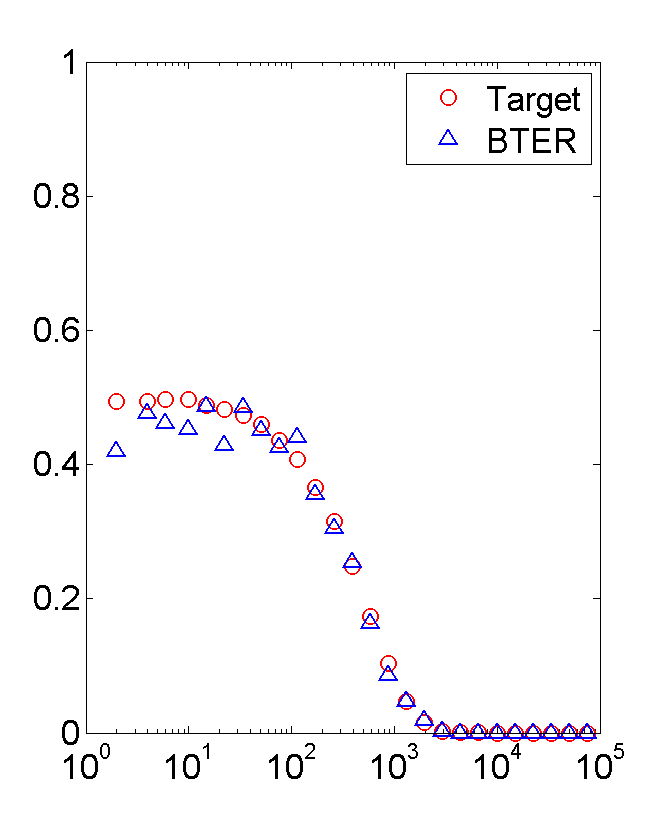}}

  \caption{Example target and actual degree distributions for 256M nodes using proposed idealized distributions.}
  \label{fig:scalability-example}
\end{figure}

\section{Conclusions and Future Work}
\label{sec:conclusions}
This paper demonstrates that the BTER generative model is useful
for modeling massive networks, especially as compared with other scalable models. 
We provide a detailed algorithm
along with analysis explaining the workings
of the method. The original paper on BTER \cite{SeKoPi12} provided
none of the implementation details and, in fact, did not directly use
the clustering coefficient data but rather estimated it via a
function. Here we give precise details on the implementation, which
is nontrivial due to issues such as repeat edges. We are able to build
a model of a graph with 120M nodes and 4.4B edges in less than 25
minutes on a 32-node Hadoop cluster. 

The development of a realistic graph model is an important step in
developing effective ``null'' models that nonetheless share the
properties of real-world networks. Such models will be useful in
detecting anomalies, statistical sampling, and community detection.
For example, the BTER model does not have larger communities beyond
the affinity blocks, whereas we might expect that real-world graphs
have a richer structure such as a hierarchy or other complex behavior.

The proposed  BTER model, along with the proposed degree and clustering coefficient distributions, may also boost benchmarking efforts in graph processing.  
The proposed degree distributions capture the essence of degree distributions that we see in practice and  generate realistic distributions even at large scales (whereas power law has a reputation of generating a few degrees that are much larger than observed in practice).  Moreover, the proposed distribution allows us to modify both the average and the maximum degree, which is critical for benchmarking.  The proposed clustering coefficient curves implicitly embed triangle structure into the graphs, which is a critical feature that distinguishes real graphs form arbitrary sparse graphs. And finally, the proposed generation algorithm scales to extremely large graphs because it generates edges in parallel.

Of course, we should list the limitations of the BTER model. First and foremost,
we only consider simple graphs. More complex models would be needed for directed
and/or weighted graphs. For instance, even capturing the degree distributions of
a directed graph is a challenge \cite{FRD-arXiv-1210.5288}. We also ignore
complex community structure, such as bipartite or near-bipartite community
structure as well as hierarchical structure, which has been observed in real-world applications \cite{Ne06,ClMoNe08}.
Such structure could potentially be incorporated by adjusting the way that affinity blocks are linked. Currently, affinity blocks link only within themselves and then randomly to other nodes in the graphs. To simulate bipartite structure, the affinity blocks could be paired. To generate hierarchical structure, the affinity blocks could be arranged in that way with excess degree being biased towards blocks that are closer in the hierarchy.
We do not incorporate node and edge types which would be relevant, for instance,
in an e-commerce graph that represents users and items, connected via purchases
and ratings of items by users. Finally, although we can generate edges in a
streaming fashion, the BTER model has no real concept of evolving parameters in time.
All of these issues are topics for future studies.

\appendix

\section{Coupon Collector Derivation}
\label{sec:coupon_collector}
Consider a universe $U$ of objects/coupons, and suppose we pick objects uniformly at random with
replacement from $U$. The following theorem proves the bound used in
equation~(\ref{eq:cc}), when $U$ is the
set of possible pairs in an affinity block (so $|U| = \binom{n_b}{2}$).
This is a simple take on the standard coupon collector problem, where we wish to pick
up all distinct coupons. (We follow the analysis of Section 3.6.1 in~\cite{MoRa}).

\begin{theorem} For a given $\rho \in (0,1)$, the expected number of independent draws required to select $\rho |U|$
distinct coupons from $U$ is $|U|\ln(1/(1-\rho)) + O(1)$.
\end{theorem}

\begin{proof} For convenience, we assume  that $\rho|U|$ is an integer. Consider a sequence of draws. Let $X_i$ (for integer $0 \leq i < \rho |U|$)
be the random variable denoting the number of draws required to get one more (distinct) coupon after $i$ distinct coupons
have been collected. Observe that the quantity of interest is $\EX[\sum_{i < \rho |U|} X_i]$, which by linearity
of expectation is $\sum_{i < \rho|U|} \EX[X_i]$. (The usual coupon collector analyses consider this sum 
for $\rho = 1$.)

When $i$ distinct coupons have already been collected, the probability that a single draw gives a new coupon is exactly $1-i/|U|$.
Think of this as probability of ``failure." The number of draws required for a success (new coupon) follows
a geometric distribution (Chap VI.8 of~\cite{Fel50}) and the mean of this is $1/(1-i/|U|) = |U|/(|U|-i)$.
Using this bound, the expected total number of draws can be expressed as follows:
\begin{align*}
	\sum_{i < \rho|U|} \EX[X_i] & = \sum_{i < \rho|U|} \frac{|U|}{|U|-i} \\
	& = |U| \Big[\sum_{i \leq |U|} \frac{1}{i} - \sum_{i \leq (1-\rho)|U|} \frac{1}{i}\Big] \\
	& = |U| \Big[ \ln |U| - \ln((1-\rho)|U|) + O(1/|U|) \Big] = |U|\ln(1/(1-\rho)) + O(1)
\end{align*}
(We use 
the standard bound for the Harmonic sum: $\sum_{i \leq r} 1/i = \ln r + \gamma + O(1/r)$, where
$\gamma$ is Euler-Mascheroni constant.)
\end{proof}

\newcommand{\BinText}[1]{(Reproduces \Fig{#1} in the original paper with log-binned plots.)}
\newcommand{\CummText}[1]{(Reproduces \Fig{#1} in the original paper with cumulative plots.)}

\section{Plots with Log-Binned Data}

Figures \ref{fig:smalldata-binned} -- \ref{fig:ideal-binned} reproduce degree distribution and clustering coefficient plots in the original paper using logarithmic binning of the degrees. The $k$th degree bin covers the range $\set{b_k,\dots,b_{k+1}-1}$ where
\begin{displaymath}
  b_k = \left\lceil \frac{\omega^{k-1}-1}{\omega-1} \right\rceil + 1.
\end{displaymath}
We use $\omega = 1.5$, which mean that the bin grows in size by 50\% at each step. For example, the first few $b_k$ values are
\begin{displaymath}
  \set{1,2,4,6,10,15,22,34,51,76,115,172,259,389,583,875,1313,1970,2955,
\dots}
\end{displaymath}
For the degree distributions, the data within each bin is summed.
For the clustering coefficients by degree, the data within each bin is averaged.

\begin{figure}[thbp]
  \centering
  \subfloat[Degree distribution for ca-HepTh]
  {\includegraphics[width=0.25\textwidth]{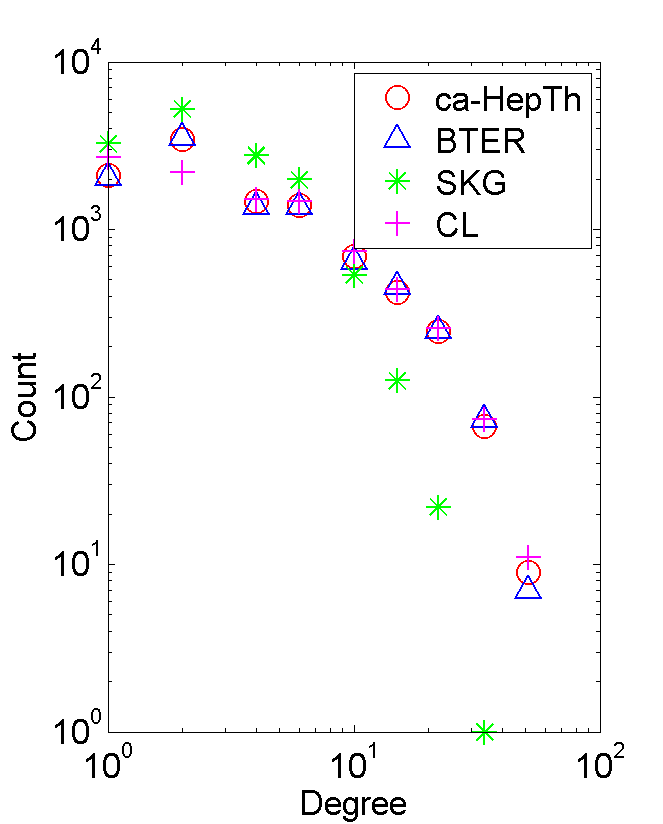}}
  \subfloat[Clustering coefficients for ca-HepTh]
  {\includegraphics[width=0.25\textwidth]{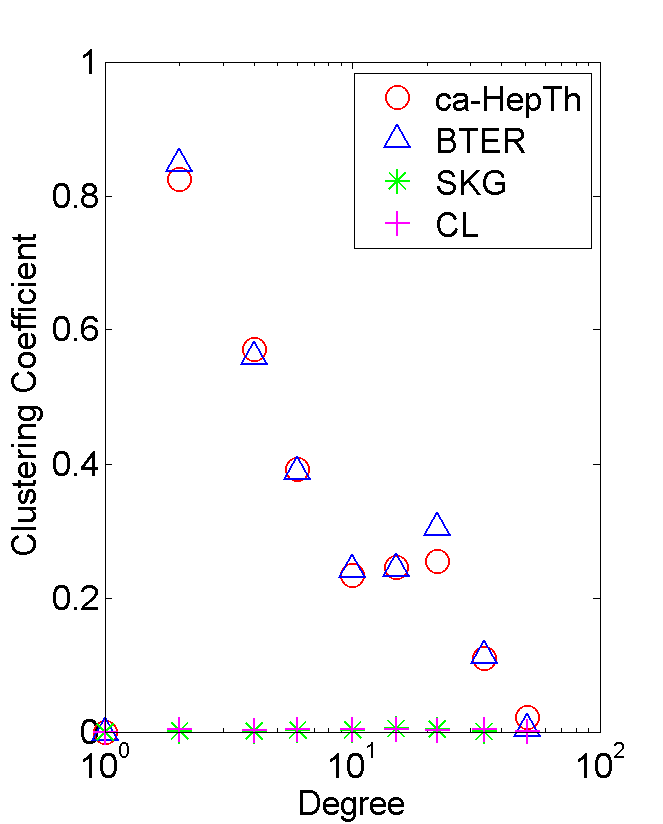}}
  \subfloat[Degree distribution for soc-Epinions1]
  {\includegraphics[width=0.25\textwidth]{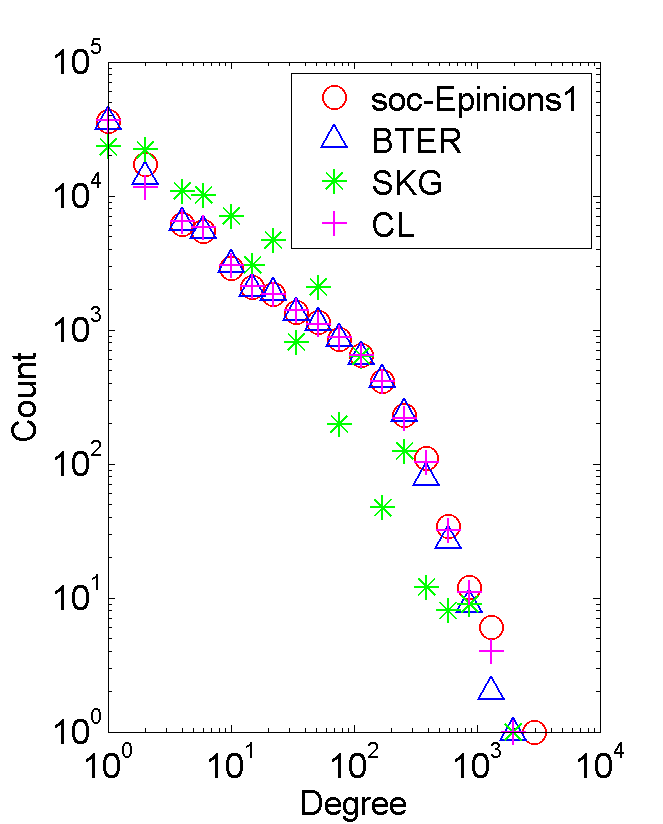}}
  \subfloat[Clustering coefficients for soc-Epinions1]
  {\includegraphics[width=0.25\textwidth]{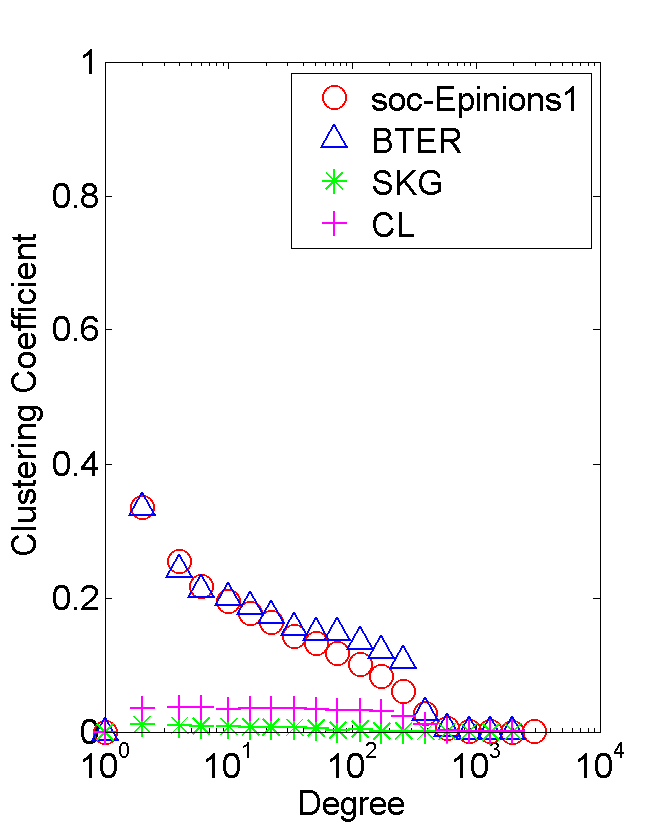}}
  \caption{Comparison of Chung-Lu (CL), SKG, and BTER on small graphs. \BinText{smalldata}}
  \label{fig:smalldata-binned}
\end{figure}

\begin{figure}[htbp]
  \centering
  \subfloat[amazon-2008]{\includegraphics[width=0.33\textwidth]{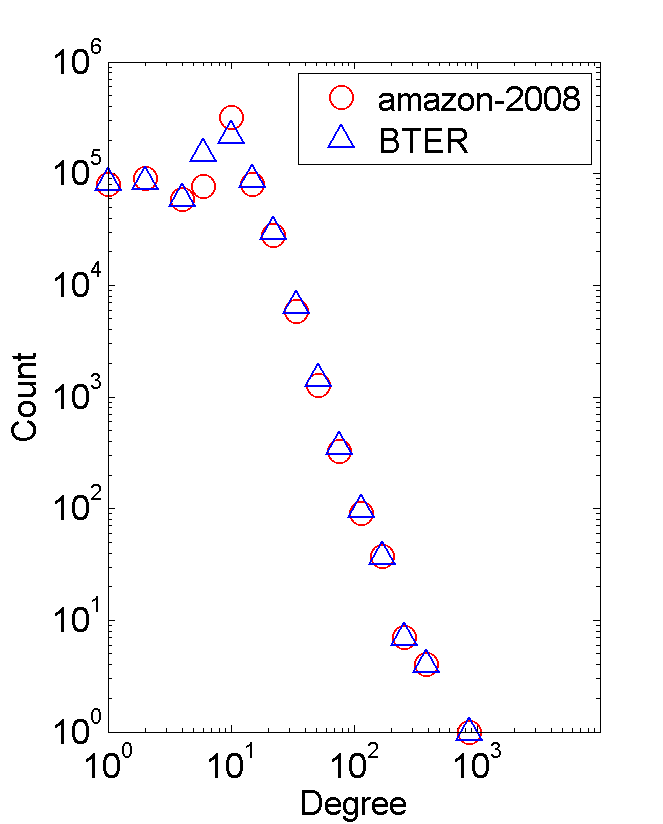}}
  \subfloat[ljournal-2008]{\includegraphics[width=0.33\textwidth]{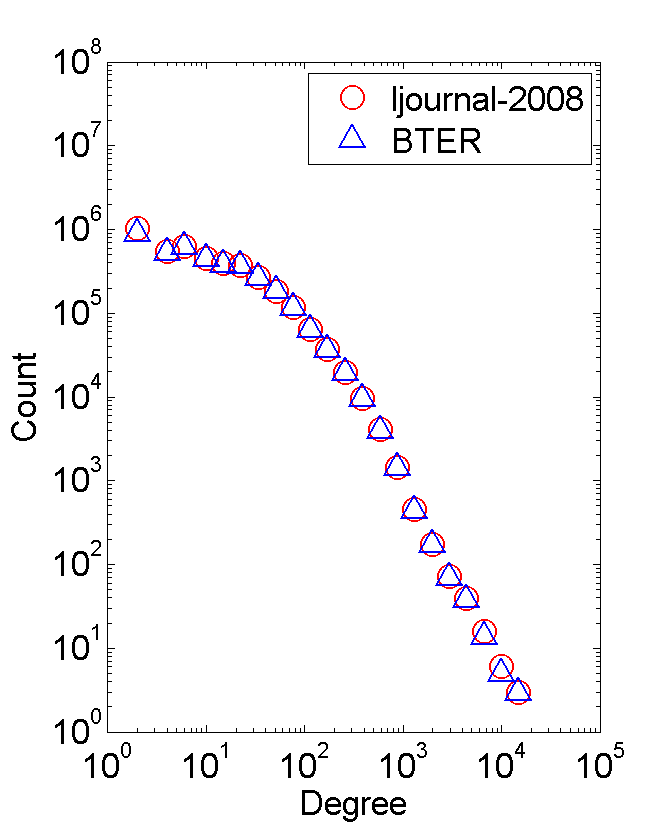}}
  \subfloat[hollywood-2011]{\includegraphics[width=0.33\textwidth]{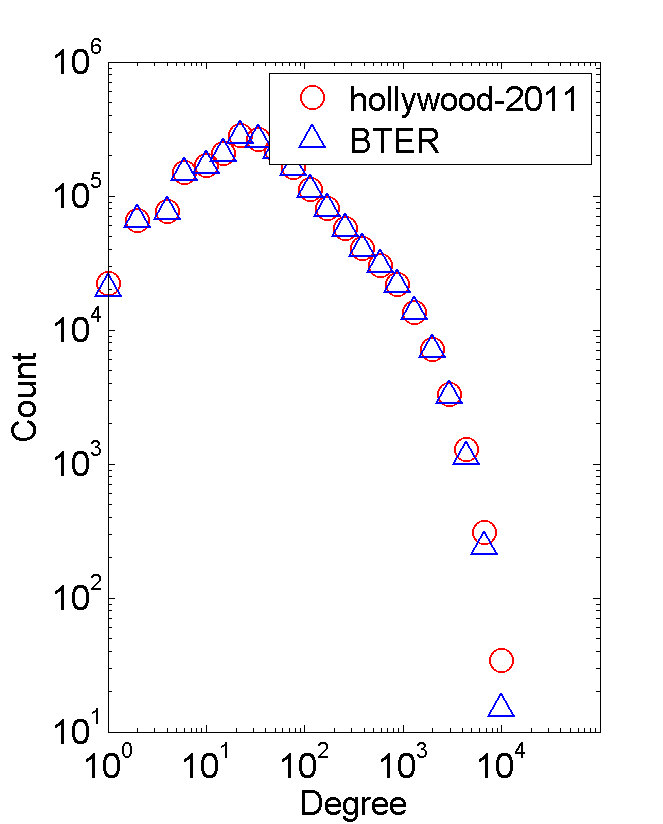}}\\
  \subfloat[twitter-2010]{\includegraphics[width=0.33\textwidth]{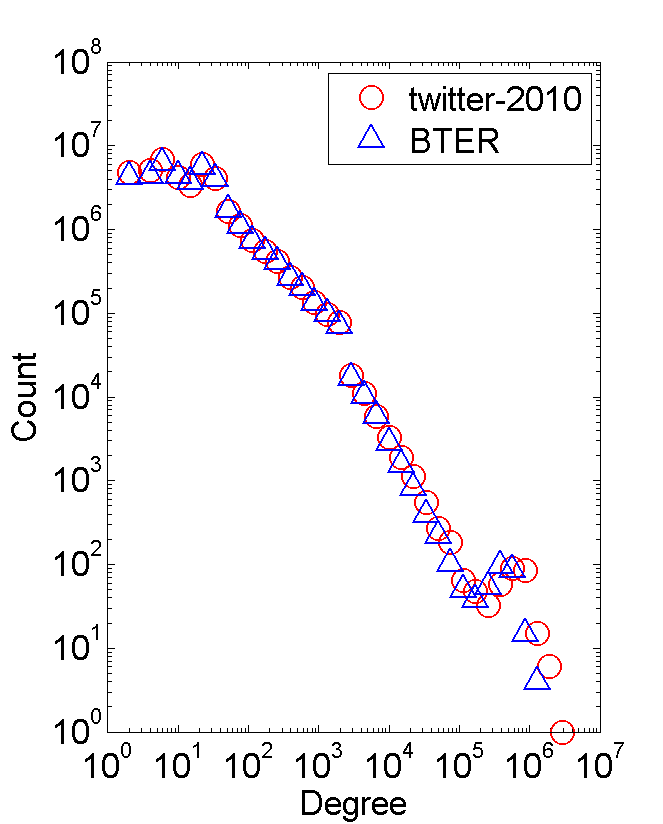}}
  \subfloat[uk-union]{\includegraphics[width=0.33\textwidth]{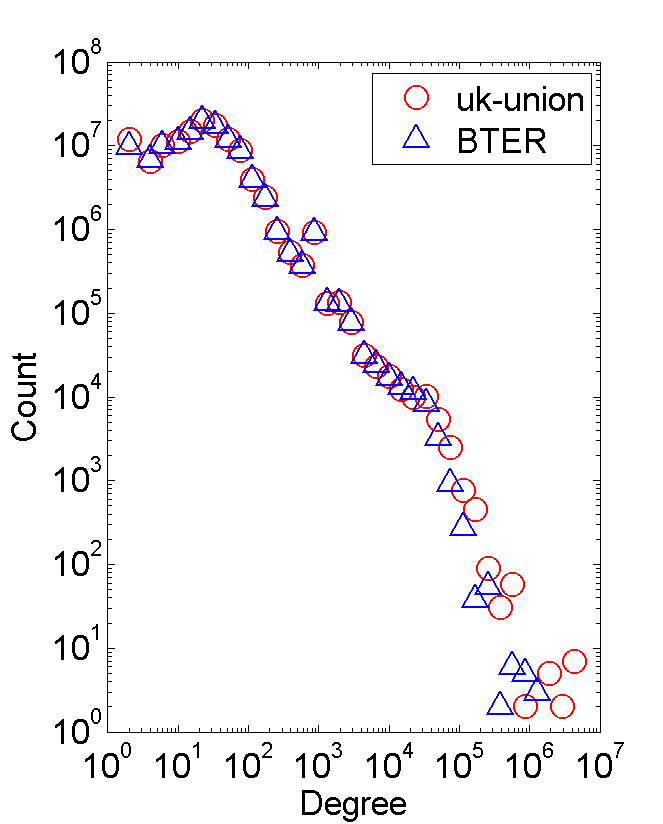}}
  \caption{Degree distributions of original and BTER-generated graphs. \BinText{degdist-realdata}}
  \label{fig:degdist-realdata-binned}
\end{figure}

\begin{figure}[htbp]
  \centering
  \subfloat[amazon-2008]{\includegraphics[width=0.33\textwidth]{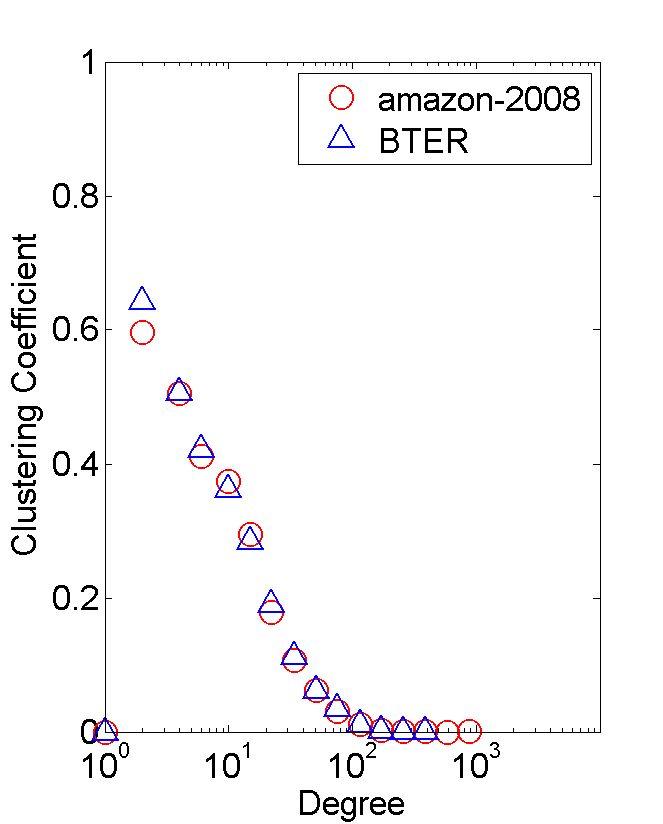}}
  \subfloat[ljournal-2008]{\includegraphics[width=0.33\textwidth]{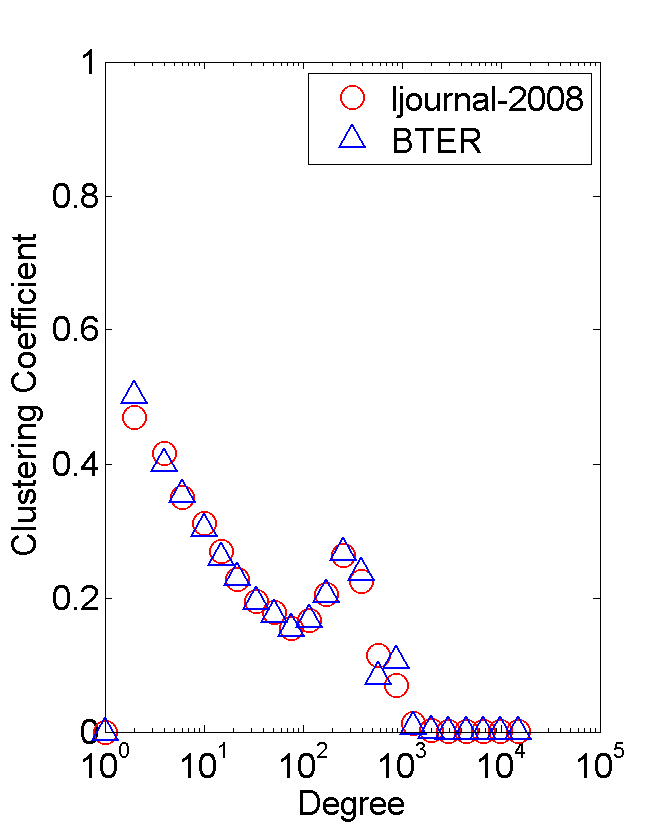}}
  \subfloat[hollywood-2011]{\includegraphics[width=0.33\textwidth]{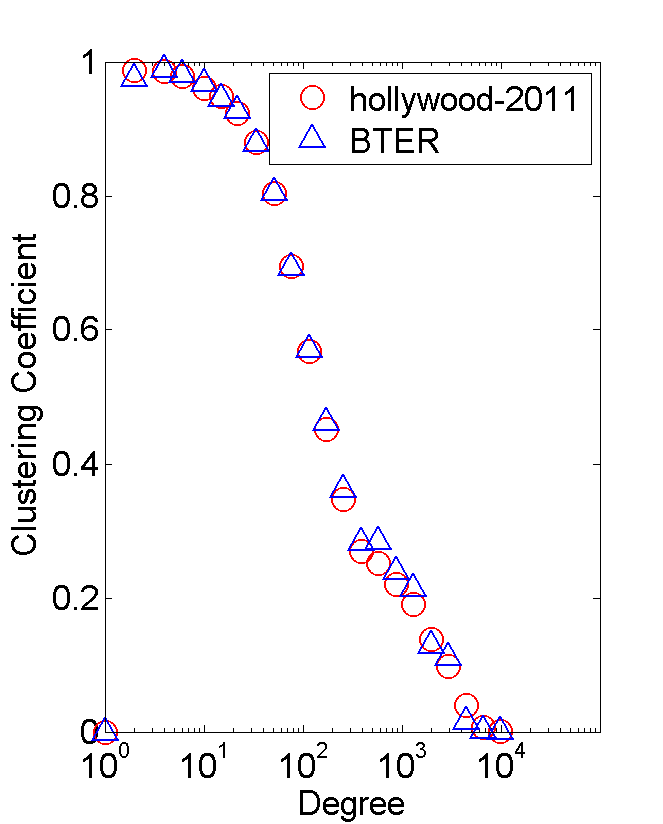}}\\
  \subfloat[twitter-2010]{\includegraphics[width=0.33\textwidth]{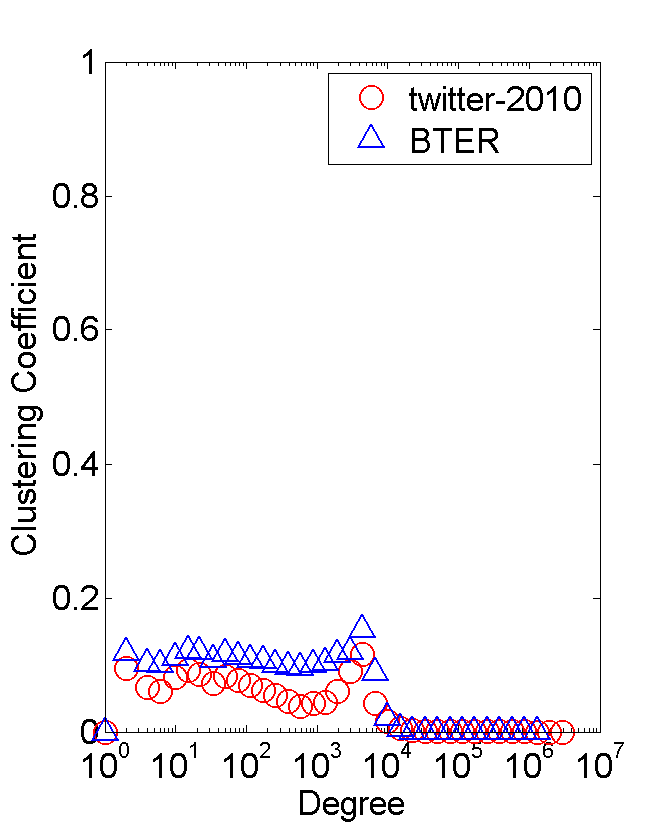}}
  \subfloat[uk-union]{\includegraphics[width=0.33\textwidth]{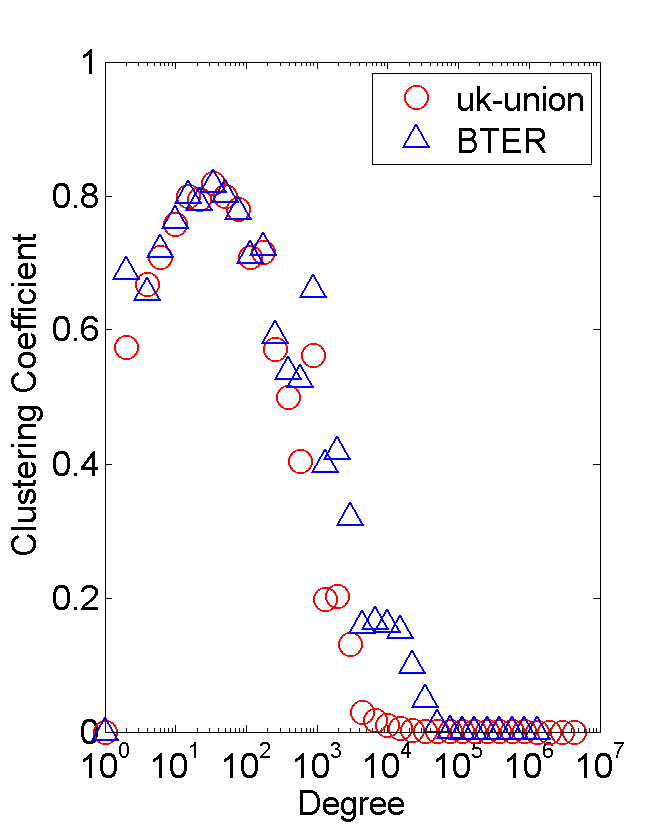}}
  \caption{Clustering coefficients of original and BTER-generated graphs. \BinText{cc-realdata}}
  \label{fig:cc-realdata-binned}
\end{figure}

\begin{figure}[htbp]
  \centering
  \subfloat[Degree distribution for Scenario 1]{\includegraphics[width=0.33\textwidth]{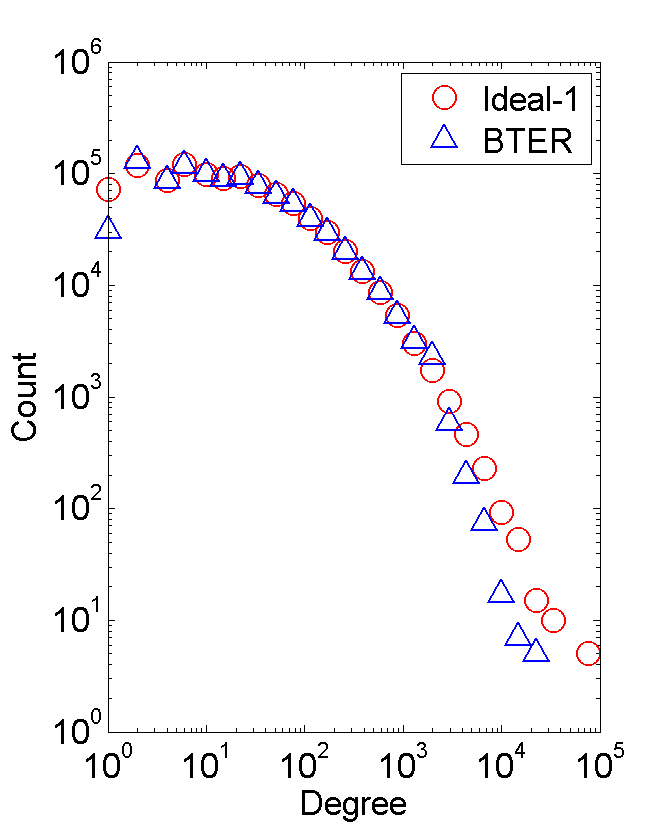}}
  \subfloat[Degree distribution for Scenario 2]{\includegraphics[width=0.33\textwidth]{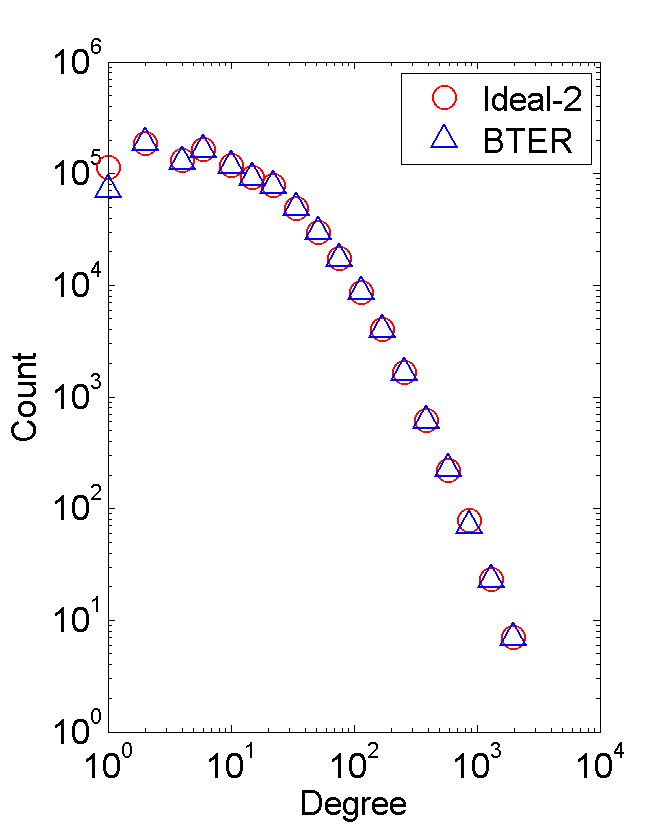}}\\
  \subfloat[Clustering coefficients for Scenario 1]{\includegraphics[width=0.33\textwidth]{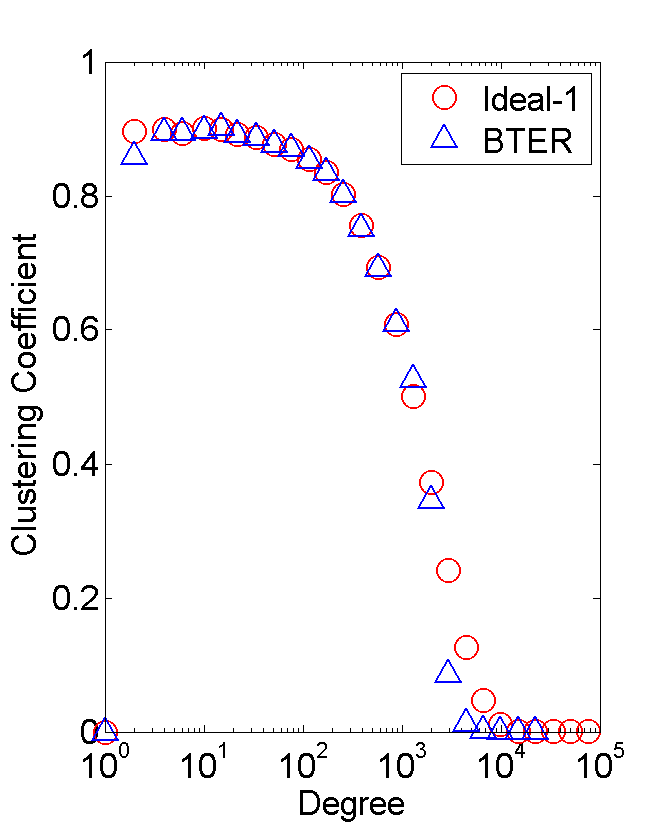}}
  \subfloat[Clustering coefficients for Scenario 2]{\includegraphics[width=0.33\textwidth]{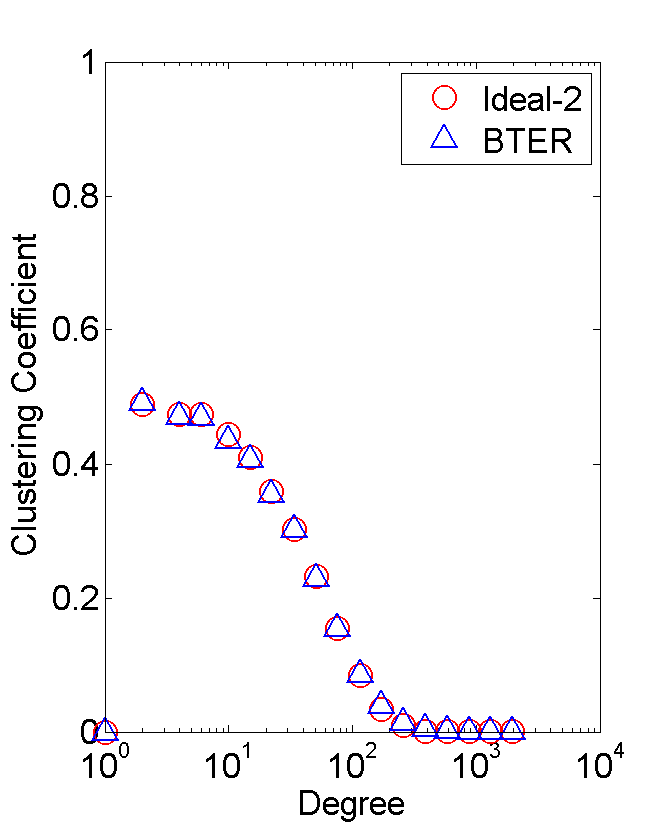}}
  \caption{Target distributions and results of BTER-generated graphs. \BinText{ideal}}
  \label{fig:ideal-binned}
\end{figure}

\section{Cumulative Degree Distributions}

Figures \ref{fig:smalldata-cumm} -- \ref{fig:ideal-cumm} reproduce degree distribution plots in the original paper using cumulative data.

\begin{figure}[thbp]
  \centering
  \subfloat[Degree distribution for ca-HepTh]
  {\includegraphics[width=0.33\textwidth]{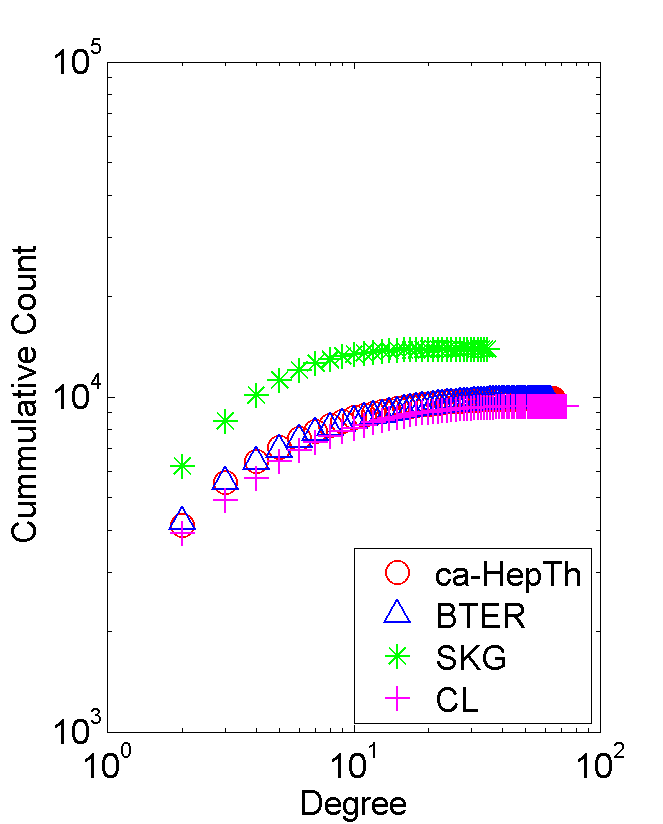}}
  \subfloat[Degree distribution for soc-Epinions1]
  {\includegraphics[width=0.33\textwidth]{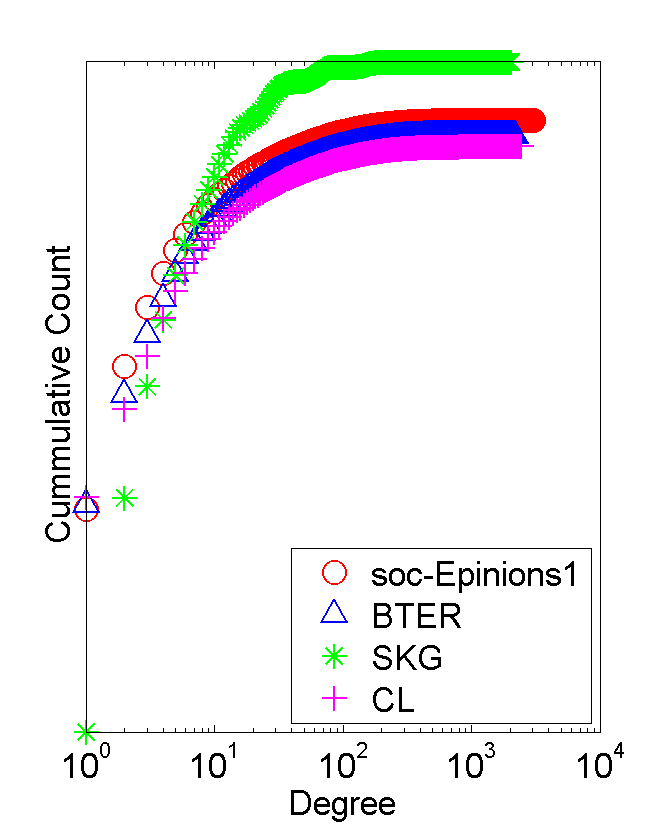}}
  \caption{Comparison of Chung-Lu (CL), SKG, and BTER on small graphs. \CummText{smalldata}}
  \label{fig:smalldata-cumm}
\end{figure}

\begin{figure}[htbp]
  \centering
  \subfloat[amazon-2008]{\includegraphics[width=0.33\textwidth]{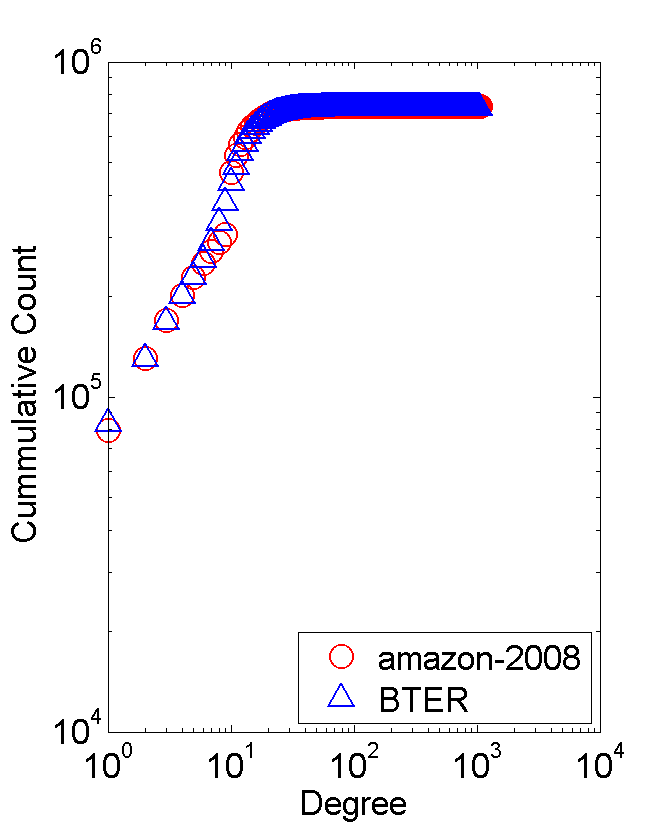}}
  \subfloat[ljournal-2008]{\includegraphics[width=0.33\textwidth]{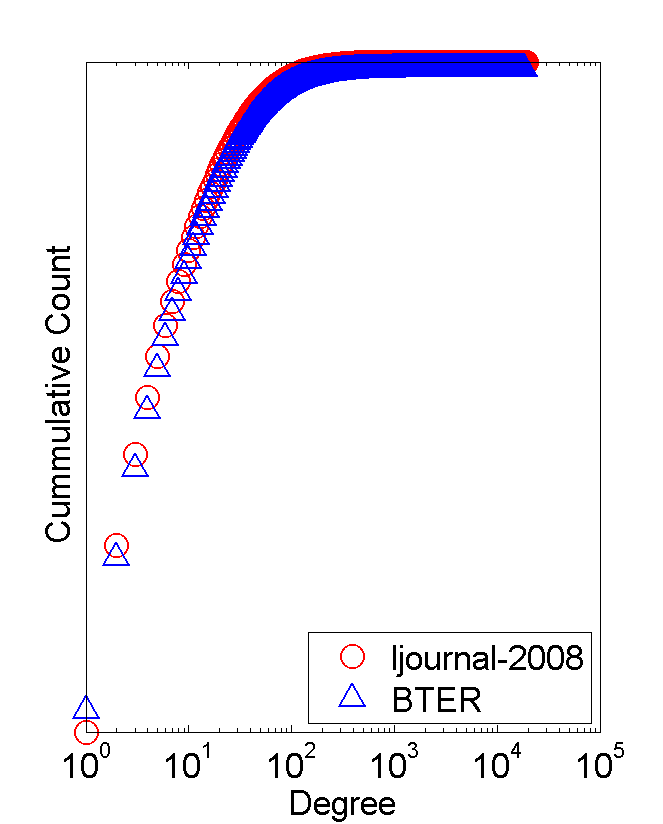}}
  \subfloat[hollywood-2011]{\includegraphics[width=0.33\textwidth]{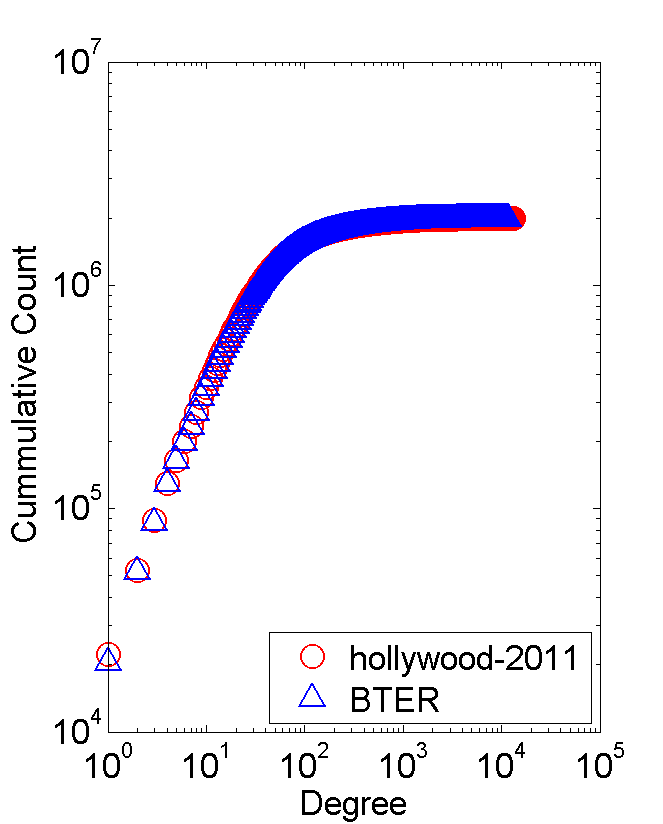}}\\
  \subfloat[twitter-2010]{\includegraphics[width=0.33\textwidth]{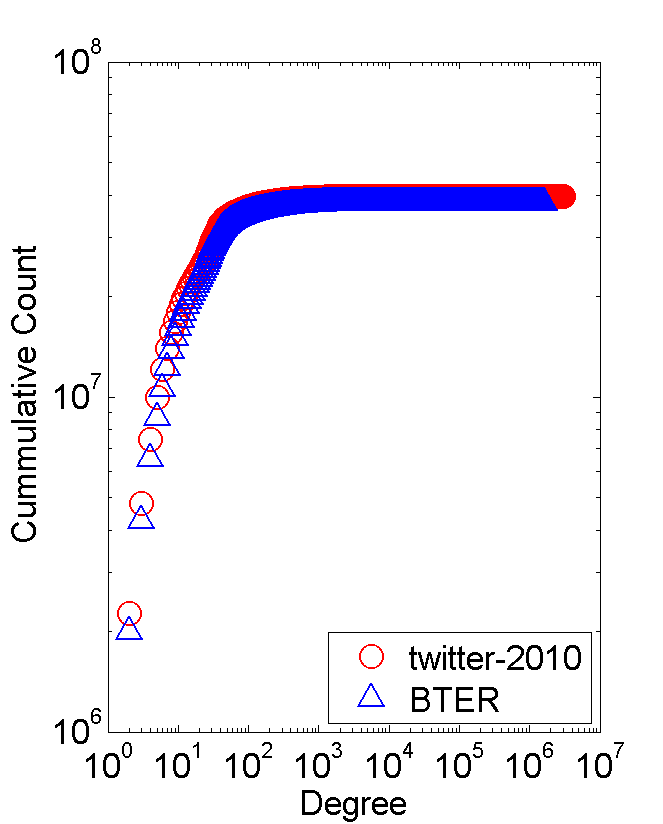}}
  \subfloat[uk-union]{\includegraphics[width=0.33\textwidth]{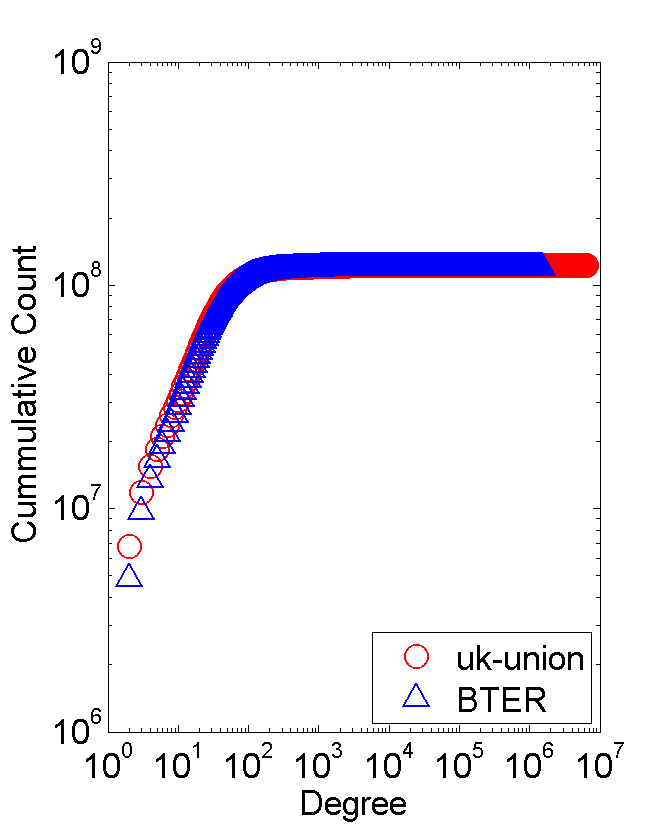}}
  \caption{Cumulative degree distributions of original and BTER-generated graphs. \CummText{degdist-realdata}}
  \label{fig:degdist-realdata-cumm}
\end{figure}

\begin{figure}[htbp]
  \centering
  \subfloat[Degree distribution for Scenario 1]{\includegraphics[width=0.33\textwidth]{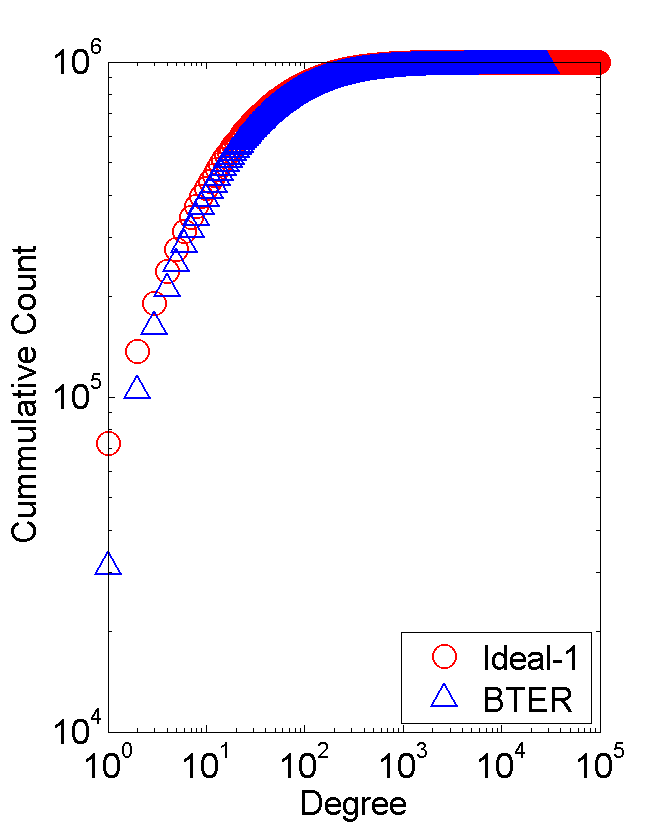}}
  \subfloat[Degree distribution for Scenario 2]{\includegraphics[width=0.33\textwidth]{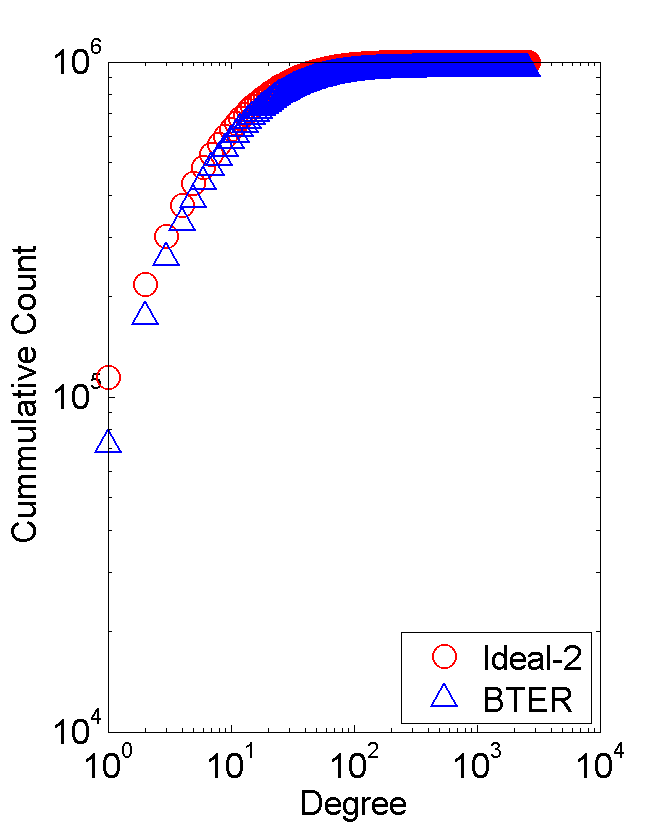}}\\
  \caption{Cumulative target distributions and results of BTER-generated graphs. \CummText{ideal}}
  \label{fig:ideal-cumm}
\end{figure}

\section{Parameter Study for Generalized Log-Normal Distribution}

In \Fig{dgln-samples}, we generate some example degree distributions using different parameters for the DGLN distribution, described in \Sec{idealdeg} or the paper. Here we use a maximum degree of $d^*=10^6$ and the number of nodes is $n=10^7$.

\begin{figure}[htbp]
  \centering
  \subfloat[$\alpha=2$ and $\delta$ varies]{\includegraphics[width=0.33\textwidth]{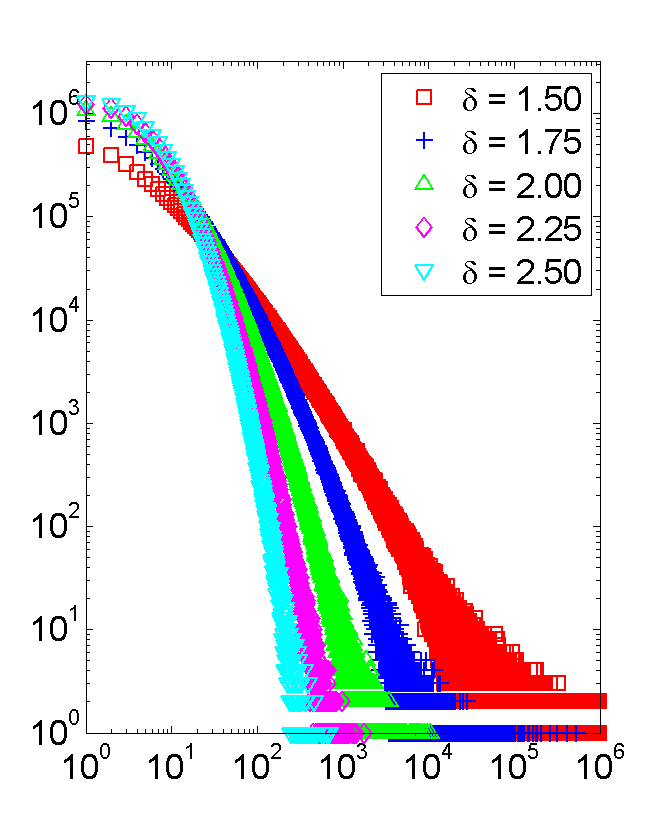}}
  \subfloat[$\alpha$ varies and $\delta=2$]{\includegraphics[width=0.33\textwidth]{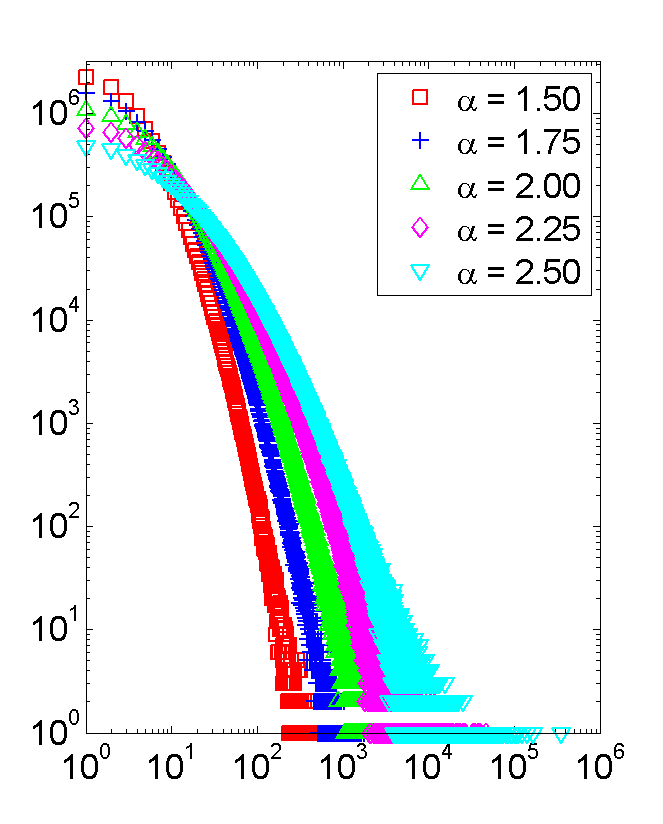}}\\
  \caption{Sample degree distributions using different parameters for DGLN.}
  \label{fig:dgln-samples}
\end{figure}

\bibliographystyle{siammod}


\end{document}